\documentclass[11pt]{article}
\usepackage[T1]{fontenc}
\usepackage[a4paper,margin=2.6cm]{geometry}
\usepackage{amsmath,amsthm,mathtools,bm,graphicx,amssymb}
\usepackage{cite}
\usepackage[hidelinks]{hyperref}
\numberwithin{equation}{section}

\newtheorem{theorem}{\bf Theorem}[section]
\newtheorem{proposition}[theorem]{\bf Proposition}
\newtheorem{corollary}[theorem]{\bf Corollary}

\newtheorem{remark}[theorem]{\bf Remark}

\newcommand{\dd}{\mathrm d}

\newcommand{\bu}{\bm u}
\newcommand{\bF}{\bm F}
\newcommand{\bP}{\bm P}
\newcommand{\bpsi}{\bm\psi}
\newcommand{\bLambda}{\bm\Lambda}
\newcommand{\bJ}{\bm J}
\newcommand{\bM}{\bm M}

\newcommand{\bc}{\bm c}
\newcommand{\bC}{\bm C}
\newcommand{\bB}{\bm B}
\newcommand{\bA}{\bm A}
\newcommand{\be}{\bm e}
\newcommand{\cT}{\mathcal T}
\newcommand{\cP}{\mathcal P}

\begin{document}

\title{Toward a Rational Extended Thermodynamics \\  of dispersive elastic media}
\author{%
Tommaso Ruggeri\\
\small Department of Mathematics, University of Bologna, Bologna, Italy\\
\small Accademia Nazionale dei Lincei, Rome, Italy
\and
Giuseppe Saccomandi\\
\small Department of Engineering, University of Perugia, Perugia, Italy\\
\small School of Mathematics, Statistics and Applied Mathematics, NUI Galway, Ireland
}
\date{}
\maketitle

\begin{abstract}
We develop a one-dimensional theory of dispersive elasticity within
Rational Extended Thermodynamics, taking local first-order balance
laws rather than higher spatial gradients as the fundamental
description. A supplementary mechanical-energy law and the
Ruggeri--Strumia main-field principle determine the admissible stress,
internal fluxes and production. A canonical two-field hierarchy is
symmetric hyperbolic under explicit convexity conditions and contains
nonlinear elasticity and a generalized-stress theory as principal
subsystems. For the reversible linear singular two-field class,
elimination of the fast stress mode yields the Love--Rosenau equation
exactly, together with a necessary and sufficient realizability
condition. Nonlinear elastic stresses and non-quadratic higher-field
energies are compatible with the same RET architecture; for the exact
nonlinear Love--Rosenau reduction we retain quadratic higher-field
inertia while allowing nonlinear elastic stress. The reduced equation
admits a travelling-wave first integral. For the leading cubic elastic
correction we obtain an exact parametric smooth solitary pulse in a
supersonic velocity window, while the truncated-cosine compacton is
excluded. Direct simulations of the hyperbolic parent system show
finite-time pulse persistence in the reversible regime and slow decay
under weak dissipation. The local parent energy flux is kept distinct
from interstitial working in the reduced theory.
\end{abstract}


\medskip
\noindent\textbf{Corresponding author:} Giuseppe Saccomandi, \href{mailto:Giuseppe.saccomandi@unipg.it}{Giuseppe.saccomandi@unipg.it}.

\section{Introduction}
\label{sec:introduction}

Real materials are characterized by one or more intrinsic length
scales that affect their mechanical behaviour at different levels of
observation.  Atomic spacing is relevant at microscopic scales, grain
size at mesoscopic scales, and inclusions, pores or secondary phases
at macroscopic scales.  In amorphous materials, structural disorder
may survive coarse graining and appear macroscopically through
localization, dispersion or an explicit dependence on internal
lengths.

Classical elasticity in an unbounded homogeneous body is
non-dispersive.  Dispersion may nevertheless be generated by geometry.
Love's theory of longitudinal waves in an elastic rod accounts for
transverse inertia and leads, in its simplest form, to an equation with
a mixed space--time derivative.  It explains why long waves in a rod
may propagate faster than short waves.  This mechanism has recently
been reconsidered from a modern mechanical viewpoint by Nobili and
Saccomandi \cite{NobiliSaccomandi2024}.

A different mechanism arises in the continualization of discrete
chains.  Retaining the first non-trivial correction associated with
the lattice spacing gives a Boussinesq-type fourth-order spatial term
\cite{Rosenau1986}.  Love-type and Boussinesq-type equations are both
dispersive, but they act on different parts of the dynamics.  The
former modifies inertia and typically preserves high-frequency
spectral stability; a truncated fourth-order spatial approximation may
instead lose stability, depending on its sign, and requires additional
boundary information.

These examples illustrate two broad routes to enriched continua.  In
the axiomatic route, the state space or the universal principles are
enlarged by introducing higher gradients, internal variables or
non-local interactions.  In the continualization route, a continuum
equation is deduced from a discrete model and retains asymptotic memory
of the underlying lattice.  Both routes have produced non-local,
peridynamic, microstructured and higher-gradient theories.  Such
models lie beyond simple materials in the classical local-history
sense \cite{TruesdellNoll2004}.

The discrete description is not automatically superior to an
axiomatic continuum theory.  A microscopic model represents one
specific mechanism, whereas an axiomatic construction identifies the
largest class of equations compatible with the chosen universal
principles.  This point is particularly relevant for dispersive
elasticity.  The phenomenological model of Rubin, Rosenau and Gottlieb
introduced dispersion through an inherent material length
\cite{RubinRosenauGottlieb1995}; its relation to dispersive simple
materials and nonlinear waves was subsequently clarified and extended
\cite{DestradeSaccomandi2006,DestradeSaccomandi2008,AmendolaMottaSaccomandiVergori2024}.

The present work follows a third, complementary route based on
Rational Extended Thermodynamics (RET).  In RET, dynamically relevant
quantities are promoted to independent fields governed by local
balance laws.  Relations containing spatial gradients arise only after
closure, elimination of internal fields or singular limiting
procedures \cite{MullerRuggeri1998,RuggeriSugiyama2015,RuggeriSugiyama2021}.
This viewpoint was formulated explicitly for apparently non-local
constitutive equations in \cite{RuggeriCan2012} and has recently been
used in nonlinear viscoelasticity, non-Newtonian fluids and
non-isothermal solid--fluid models
\cite{RuggeriVisco2024,RuggeriNonNewtonian2025,ArimaRuggeri2026}.

Our mechanical target is the one-dimensional Love--Rosenau family
\begin{equation}
 \rho u_{tt}-\mu u_{zz}+\alpha u_{zzzz}-\beta u_{zztt}=0,
 \qquad \rho,\mu>0,
 \label{eq:linear-target}
\end{equation}
where \(u=u(z,t)\) denotes the longitudinal displacement. For a plane
wave
\begin{equation}
 u(z,t)=\widehat u\,e^{i(\xi z-\omega t)},
 \label{eq:linear-target-plane-wave}
\end{equation}
where \(\xi\) is the wavenumber and \(\omega\) is the angular frequency,
one has
\[
 u_{tt}=-\omega^2u,\qquad
 u_{zz}=-\xi^2u,\qquad
 u_{zzzz}=\xi^4u,\qquad
 u_{zztt}=\xi^2\omega^2u.
\]
Substitution into \eqref{eq:linear-target} gives
\[
 (\rho+\beta\xi^2)\omega^2=\mu\xi^2+\alpha\xi^4,
\]
and hence the dispersion relation
\begin{equation}
 \omega^2=
 \frac{\mu\xi^2+\alpha\xi^4}{\rho+\beta\xi^2}.
 \label{eq:linear-target-dispersion}
\end{equation}
Rather than treating \eqref{eq:linear-target} as fundamental, we ask
whether it can be generated by a local first-order RET hierarchy with
a convex supplementary energy and a symmetric-hyperbolic main-field
form.  We first establish the compatibility conditions imposed by the
supplementary energy law, then analyse the linear theory and its
realizability domain, and only
subsequently pass to nonlinear energies and travelling waves. We
finally compare the exact pulse of the singular reduced equation with
direct simulations of the full hyperbolic parent hierarchy for small
positive fast inertia, both in the reversible case and under weak
dissipation. The local internal energy flux of the parent system is
kept conceptually distinct from the interstitial-working term that
appears in the reduced higher-gradient equation.

\section{RET postulates, balance laws and the issue of objectivity}
\label{sec:postulates-objectivity}

Rational Extended Thermodynamics originated from the hierarchy of
moments in kinetic theory. In that setting, the flux of a retained
moment is itself a higher-order moment. When such a flux remains
dynamically relevant on the time scale of observation, it is promoted
to an independent field and is governed by a further balance law.
Repeating this procedure generates the recursive hierarchy of moment
equations characteristic of kinetic theories
\cite{MullerRuggeri1998,RuggeriSugiyama2015,
RuggeriSugiyama2021}.

Although this precise recursive structure is specific to moment
theories, its underlying methodological principle can be used more
generally. A quantity which enters the theory as a new flux, or more
generally as a dynamically relevant non-equilibrium quantity, is
promoted to an independent field governed by an additional balance
law. The density, flux and production entering that law are taken to
be local functions of the enlarged set of state variables. They are
not assigned through an instantaneous non-local constitutive equation
containing spatial gradients \cite{RuggeriCan2012}.

According to this viewpoint, equations such as Fourier's law, the
Navier--Stokes stress relation, Fick's law and Darcy's law are not
regarded as fundamental constitutive equations of the 
theory. They may instead arise after rapidly relaxing fields have been
eliminated, through a Maxwellian iteration, or in a suitable singular
limit. The local first-order balance hierarchy is the parent system,
whereas the familiar gradient-dependent or parabolic equation is a
reduced description.

This interpretation has recently been applied outside the traditional
kinetic setting. In nonlinear viscoelasticity, the additional stress
was promoted to an independent field governed by a balance law, whose
admissible density, flux and production were determined through a
supplementary energy principle \cite{RuggeriVisco2024}. The same
methodology was subsequently used for non-Newtonian fluids with finite
relaxation time and for a non-isothermal unified theory of
viscoelastic solids and non-Newtonian fluids
\cite{RuggeriNonNewtonian2025,ArimaRuggeri2026}.

\subsection{Structural RET requirements}

Following this viewpoint, we adopt the following structural
requirements.

\begin{enumerate}
\item[(i)]
Every additional quantity which remains dynamically relevant on the
observational time scale is treated as an independent field governed
by a balance law. The corresponding balance density, flux and
production are local functions of the enlarged state variables.
Consequently, after a choice of independent state variables, the
resulting local first-order parent hierarchy is a quasilinear system.

\item[(ii)]
The enlarged hierarchy possesses a quasi-linear supplementary thermodynamic law
which is compatible with all the balance equations through a common
set of multipliers.

\item[(iii)]
These multipliers constitute the main field
\(\bu'\). The corresponding Legendre potentials generate the balance
densities and fluxes.

\item[(iv)]
The residual production in the supplementary law has the sign required
by the second law. In the present isothermal mechanical theory, the
supplementary mechanical-energy balance takes the place of the entropy
balance, and its residual production is non-positive.

\item[(v)]
The supplementary density satisfies the thermodynamic stability
requirement. According to the adopted sign convention, either the
entropy density or its negative is strictly convex with respect to the
balance densities. In the present energy formulation, the
supplementary energy density is convex.

\item[(vi)]
The main field provides symmetrizing variables. The matrices of the
homogeneous system are Hessians of the RET potentials and are therefore
symmetric, while strict convexity makes the temporal matrix positive
definite. Consequently, the homogeneous parent system is symmetric
hyperbolic.

\item[(vii)]
According to the definition of Boillat and Ruggeri
\cite{BoillatRuggeri1997}, principal subsystems are obtained by fixing
some components of the main field and deleting the corresponding balance
laws. In this way RET possesses a natural structure of nesting theories.
The general theory ensures that every principal subsystem inherits the
entropy principle, with a convex subentropy, and that its characteristic
eigenvalues satisfy the corresponding subcharacteristic conditions,
namely they interlace those of the parent system.

\end{enumerate}

\subsection{Field equations and material frame indifference}

The locality assumption in item~(i) must not be confused with the
principle of material frame indifference. Before this principle is
applied, one must first establish whether the equation under
consideration is genuinely constitutive. Material frame indifference
is a restriction on constitutive mappings; it is not a requirement to
be imposed on field equations themselves.

This distinction is particularly important in Extended Thermodynamics.
As recognized independently by Bressan and Ruggeri
\cite{Bressan1982,RuggeriPadova1982}, when heat flux, stress, or another
non-equilibrium quantity is promoted to the status of an independent
field, its evolution equation is a genuine field equation, on the same
footing as the classical balance laws of continuum mechanics. Such an
equation is therefore not itself subject to material frame
indifference. In the present classical setting, the complete system of
field equations must instead possess the appropriate transformation
properties under changes of observer: in particular, it must satisfy
Galilean invariance under transformations between inertial frames,
whereas in non-inertial frames the balance equations acquire the
corresponding inertial terms. In the one-dimensional Lagrangian
specialization used below, a constant Galilean boost changes
\(v\mapsto v-V\) while leaving the material coordinate, deformation
and internal variables unchanged; since a constant shift leaves
\(v_t\) and \(v_z\) unchanged, the mechanical field equations retain
their form.

The local mappings entering these balance laws, on the other hand,
are genuine constitutive equations and must satisfy material frame
indifference. Thus the principle is not abandoned in RET; rather, it is
applied only after the constitutive character of the relation has been
established. Within the standpoint adopted here, the fundamental
constitutive equations are precisely the local mappings of the parent
balance hierarchy \cite{RuggeriCan2012,RuggeriVisco2024}.

M\"uller had already pointed out the frame dependence of stress and
heat flux and the limitations of an indiscriminate application of
material frame indifference to transport theories derived from kinetic
considerations \cite{Muller1972,Soderholm1976}. From the present
perspective, familiar gradient-dependent relations such as the
Fourier, Navier--Stokes, Fick and Darcy laws are not regarded as
fundamental constitutive mappings. They can instead emerge as reduced
field equations from a larger local hierarchy, for instance through a
Maxwellian iteration or an appropriate relaxation limit
\cite{RuggeriCan2012}.

The same distinction is relevant for the evolution equation of an
independent stress field. One should not enforce material frame
indifference on the balance equation itself by replacing its time
derivative, term by term, with an upper-convected, lower-convected or
corotational derivative. Within the RET approach, the transformation
properties are to be verified for the complete system, while material
frame indifference is imposed on the genuine local constitutive
mappings. An objective derivative introduced solely for the purpose of making
the stress evolution equation formally objective is therefore
unnecessary within the present formulation
\cite{RuggeriVisco2024}.

\subsection{Main field, dual potentials and symmetric form}

We now recall the general main-field structure. Let
\[
 x^\alpha=(x^0,x^1,x^2,x^3)
          =(t,x^1,x^2,x^3),
 \qquad
 \partial_\alpha=\frac{\partial}{\partial x^\alpha},
\]
where Greek indices range over
\(\alpha=0,1,2,3\), and repeated Greek indices are summed. This is only
a compact space--time notation for a classical theory and does not
imply a relativistic interpretation.

A general system of balance laws can be written as
\begin{equation}
 \partial_\alpha\bF^\alpha(\bu)
 =
 \bP(\bu).
 \label{eq:general-balance-form}
\end{equation}
Here \(\bF^0(\bu)\) is the vector of balance densities,
\(\bF^i(\bu)\), \(i=1,2,3\), are the spatial fluxes, and
\(\bP(\bu)\) is the production vector. Explicitly,
\begin{equation}
 \partial_t\bF^0(\bu)
 +
 \sum_{i=1}^{3}
 \partial_{x^i}\bF^i(\bu)
 =
 \bP(\bu).
 \label{eq:general-balance-expanded}
\end{equation}

Suppose that the balance hierarchy admits a supplementary law
\begin{equation}
 \partial_\alpha h^\alpha(\bu)
 =
 \Sigma(\bu),
 \label{eq:general-supplementary-law}
\end{equation}
or equivalently,
\begin{equation}
 \partial_t h^0(\bu)
 +
 \sum_{i=1}^{3}
 \partial_{x^i}h^i(\bu)
 =
 \Sigma(\bu).
 \label{eq:general-supplementary-expanded}
\end{equation}
In the present isothermal energy formulation,
\(h^0\) is the supplementary mechanical-energy density,
\(h^i\) are the corresponding spatial energy fluxes, and
\(\Sigma\leq0\).

Ruggeri and Strumia \cite{RuggeriStrumia1981} observed that both the
balance system \eqref{eq:general-balance-form} and the supplementary law
\eqref{eq:general-supplementary-law} are quasilinear, namely linear in
the first derivatives of the fields. Therefore, under the usual
regularity and rank assumptions, the requirement that every solution
of \eqref{eq:general-balance-form} satisfy
\eqref{eq:general-supplementary-law} identically is expressed by the
existence of a state-dependent covector field
\[
 \bu'=\bu'(\bu).
\]
Throughout the paper, \(\bu'\) is represented as a row vector, whereas
\(\bF^\alpha\), \(\bP\), and the balance-density vector are column
vectors. Thus left multiplication of the balance system by \(\bu'\)
produces a scalar equation, namely the supplementary law; see also
\cite{MullerRuggeri1998}.

The compatibility conditions between the balance hierarchy and the
supplementary law are
\begin{equation}
 \dd h^\alpha
 =
 \bu'\,\dd\bF^\alpha,
 \qquad
 \alpha=0,1,2,3,
 \qquad
 \Sigma=\bu'\,\bP.
 \label{eq:general-main-field-relations}
\end{equation}
In particular,
\[
 \dd h^0=\bu'\,\dd\bF^0,
\]
so that \(\bu'\) is thermodynamically conjugate to the column
vector of balance densities.

Introducing the dual RET potentials
\begin{equation}
 h'^\alpha
 =
 \bu'\,\bF^\alpha-h^\alpha,
 \qquad
 \alpha=0,1,2,3,
 \label{eq:general-dual-potentials}
\end{equation}
one obtains
\begin{equation}
 \dd h'^\alpha
 =
 \dd\bu'\,\bF^\alpha,
 \qquad
 \alpha=0,1,2,3.
 \label{eq:general-potential-differential}
\end{equation}
Writing \(u'_A\) and \(F_A^\alpha\) for the components of the row
covector and column flux vector, respectively, one has
\begin{equation}
 F_A^\alpha
 =
 \frac{\partial h'^\alpha}{\partial u'_A},
 \qquad
 A^\alpha_{AB}
 =
 \frac{\partial F_A^\alpha}{\partial u'_B}
 =
 \frac{\partial^2h'^\alpha}
 {\partial u'_A\partial u'_B}.
 \label{eq:general-Hessian-relations}
\end{equation}
Hence the matrices
\(\bA^\alpha=(A^\alpha_{AB})\) are symmetric Hessian matrices. Therefore,
in terms of the column vector \((\bu')^T\), the balance system assumes
the Godunov form \cite{Godunov1961},
\begin{equation}
 \bA^0\,\partial_t(\bu')^T
 +
 \sum_{i=1}^{3}
 \bA^i\,\partial_{x^i}(\bu')^T
 =
 \bP\bigl(\bu(\bu')\bigr),
 \qquad
 \bA^\alpha
 =
 \left(
 \frac{\partial^2h'^\alpha}
 {\partial u'_A\partial u'_B}
 \right)_{A,B}.
 \label{eq:main-field-system-3D}
\end{equation}
Godunov introduced this potential symmetric form for an important class
of quasilinear systems including the Euler equations of fluid dynamics
and systems arising from variational principles.

At this point the terminology introduced by Ruggeri and Strumia
\cite{RuggeriStrumia1981} becomes transparent. The row covector
\(\bu'\) is called the \emph{main field} precisely because, when it is
taken as the field of independent variables, the original quasilinear
balance system is transformed into the symmetric form
\eqref{eq:main-field-system-3D}. The Ruggeri--Strumia construction is
covariant and is therefore directly suited also to relativistic
theories.

In the classical case, after selecting a time direction and taking the
balance densities \(\bF^0\) as independent variables, the construction
reduces to the field first introduced by Boillat in 1974
\cite{Boillat1974},
\[
 \bu'=\frac{\partial h^0}{\partial \bF^0}.
\]
This purely temporal definition is not intrinsically covariant in the
relativistic case, because \(h^0\) and \(\bF^0\) are temporal components
of spacetime quantities and are not scalars. The Ruggeri--Strumia
formulation overcomes this limitation by defining the main field through
the full covariant compatibility relations
\eqref{eq:general-main-field-relations}.

Since
\[
 h'^0=\bu'\,\bF^0-h^0
\]
is the Legendre transform of \(h^0\) with respect to the balance-density
vector \(\bF^0\), the Ruggeri--Strumia stability requirement is the
strict convexity of \(h'^0\) as a function of the main field \(\bu'\).
In differential form this condition reads
\begin{equation}
 \dd\bu'\,\dd\bF^0
 =
 \dd\bu'\,\bA^0(\dd\bu')^T
 >0
 \qquad
 \text{for every }\dd\bu'\neq0.
 \label{eq:RS-convexity-general}
\end{equation}
Equivalently,
\begin{equation}
 \bA^0
 =
 \left(
 \frac{\partial^2h'^0}
 {\partial u'_A\partial u'_B}
 \right)_{A,B}
 \qquad \text{is positive definite}.
 \label{eq:temporal-potential-convexity}
\end{equation}
Hence the temporal matrix in \eqref{eq:main-field-system-3D} is positive
definite, and the homogeneous part of the system is symmetric
hyperbolic.

The relation between a convex supplementary density and a symmetric
form of a system of conservation laws goes back to the theory of convex
extensions developed by Friedrichs and Lax
\cite{FriedrichsLax1971}. In nonlinear mechanics, the corresponding
main-field structure and its relation to entropy growth across shocks
were established by Boillat and Ruggeri
\cite{BoillatRuggeri1980}. The construction developed below extends
this structure from nonlinear elasticity to a hierarchy containing two
additional internal balance fields.

Although the preceding formulation has been given in three space
dimensions, the model studied in the remainder of the paper is
one-dimensional. We set \(x^1=z\), assume that all fields are
independent of \(x^2\) and \(x^3\), and write
\begin{equation}
 \partial_t\bF^0(\bu)
 +
 \partial_z\bF^1(\bu)
 =
 \bP(\bu),
 \label{eq:one-dimensional-balance-form}
\end{equation}
together with
\begin{equation}
 \partial_t h^0(\bu)
 +
 \partial_z h^1(\bu)
 =
 \Sigma(\bu).
 \label{eq:one-dimensional-supplementary-law}
\end{equation}
All the main-field, potential and convexity relations derived above
remain valid after this one-dimensional specialization.

We finally note that the RET main field should not in general be
identified with the Lagrange multipliers of Liu's entropy-principle
procedure \cite{Liu1972}. In the general Liu procedure, the multipliers
need not constitute a field, that is, a complete set of independent
variables with as many components as there are equations in the balance
system. In the present setting, however, both the governing equations
and the supplementary law are in quasilinear balance form. The RET main
field \(\bu'\) therefore constitutes a particular set of
Lagrange--Liu multipliers which, in addition, forms a complete set of
independent variables and, by the strict convexity of the dual potential
\(h'^0\), is globally in one-to-one correspondence with the vector of
balance densities \(\bF^0\) on the admissible state domain.

\section{One-dimensional dispersive elasticity and its local balance hierarchy}
\label{sec:elastic-balance-hierarchy}

We now apply the general RET structure of
Section~\ref{sec:postulates-objectivity} to one-dimensional
elastodynamics. The aim is to construct a local first-order parent
system whose basic mechanical part is classical nonlinear elasticity
and whose singular reduction produces a dispersive
Love--Rosenau-type equation.

We consider longitudinal motions of a homogeneous elastic medium in a
Lagrangian description. As above, \(z\) denotes the material coordinate
and \(u=u(z,t)\) the longitudinal displacement. The velocity and deformation
gradient are introduced as the kinematic fields
\begin{equation}
 v=u_t,
 \qquad
 F=u_z.
 \label{eq:vF}
\end{equation}
Their definitions imply the exact compatibility equation
\begin{equation}
 F_t-v_z=0.
 \label{eq:kinematic}
\end{equation}

In classical one-dimensional nonlinear elasticity, \(v\) and \(F\)
are sufficient to describe the motion, and the stress is prescribed by
an elastic constitutive function. In the extended theory, dispersion
is instead generated by adjoining dynamically independent internal
fields to this elastic structure. These fields are not identified a
priori with spatial gradients of \(F\) or \(v\); they are governed by
their own local balance laws, in accordance with the RET principles of
Section~\ref{sec:postulates-objectivity}.

At the general level we introduce two additional independent fields
\(P\) and \(Q\). No identification such as \(P=F_z\), \(Q=v_z\), or
\(P=\) viscous stress is imposed. Their physical interpretation will
be specialized only after the supplementary-law compatibility
conditions have been established. In the canonical hierarchy below,
the first field becomes a generalized stress and \(Q\) represents the
next dynamical internal level coupled to it; in the quadratic
specialization this higher mode carries the finite internal inertia
responsible for dispersion.

The balance of momentum is written as
\begin{equation}
 \rho v_t-\partial_z\cT(F,P,Q)=0,
 \qquad \rho>0,
 \label{eq:momentum}
\end{equation}
where the total stress \(\cT\) is unknown.

The two additional balance laws are taken in the local form
\begin{subequations}
\label{eq:additional-laws}
\begin{align}
 &\partial_t\psi_1(F,P,Q)
 +\partial_z\Omega_1(F,P,Q)
 =\cP_1(F,P,Q),
 \label{eq:additional-1}\\
 &\partial_t\psi_2(F,P,Q)
 +\partial_z\Omega_2(F,P,Q)
 =\cP_2(F,P,Q).
 \label{eq:additional-2}
\end{align}
\end{subequations}
Because the description is Lagrangian, no explicit velocity
dependence is introduced in \(\Omega_1\) and \(\Omega_2\). The spatial
velocity gradient is already present through \(F_t=v_z\).

The column vector of balance densities is
\begin{equation}
 \bu=
 \bigl(
 \rho v,\,
 F,\,
 \psi_1(F,P,Q),\,
 \psi_2(F,P,Q)
 \bigr)^T.
 \label{eq:density-vector}
\end{equation}
Assume that
\begin{equation}
 (F,P,Q)\longmapsto
 \bigl(F,\psi_1(F,P,Q),\psi_2(F,P,Q)\bigr)
 \label{eq:invertibility-map}
\end{equation}
is locally invertible. Then
\begin{equation}
 \partial_t\bu+\partial_z\bF(\bu)=\bP(\bu),
 \label{eq:balance-vector}
\end{equation}
with the column flux and production vectors
\begin{equation}
 \bF=
 \bigl(-\cT,-v,\Omega_1,\Omega_2\bigr)^T,
 \qquad
 \bP=
 \bigl(0,0,\cP_1,\cP_2\bigr)^T.
 \label{eq:flux-production}
\end{equation}

\subsection{Supplementary energy law}

In accordance with the RET requirements stated in
Section~\ref{sec:postulates-objectivity}, the enlarged elastic
hierarchy is required to possess a supplementary dissipative energy
law,
\begin{equation}
 \partial_t h^0+\partial_z h^1=\Sigma,
 \qquad
 \Sigma\leq0.
 \label{eq:supplementary-law}
\end{equation}

The supplementary density is assumed to have the general form
\begin{equation}
 h^0(v,F,P,Q)
 =
 \frac{\rho}{2}v^2+\Psi(F,P,Q),
 \label{eq:general-energy}
\end{equation}
where \(\Psi\) is the internal energy associated with the deformation
and the two additional fields. At this stage no additive decomposition
of \(\Psi\) is assumed.

The supplementary energy flux is written as
\begin{equation}
 h^1(v,F,P,Q)
 =
 -v\,\cT(F,P,Q)+k(F,P,Q).
 \label{eq:general-energy-flux}
\end{equation}
The first contribution is the mechanical power flux associated with
the total stress. The scalar \(k\) is an additional energy flux carried
by the internal degrees of freedom. It is not a heat flux, since the
present theory is isothermal and contains no heat conduction. Its role
is analogous, at the structural level, to the M\"uller extra entropy
flux introduced in general non-isothermal thermodynamics
\cite{Muller1967}. In the classical Coleman--Noll formulation
\cite{ColemanNoll1963}, the entropy flux has the standard form
\(\boldsymbol{\Phi}=\mathbf{q}/T\). M\"uller relaxed this restriction by
allowing the entropy flux to be a general constitutive quantity, which
may equivalently be written as
\[
 \boldsymbol{\Phi}
 =
 \frac{\mathbf{q}}{T}
 +
 \boldsymbol{\Phi}^{\mathrm{ex}},
\]
where \(\boldsymbol{\Phi}^{\mathrm{ex}}\) denotes the extra entropy flux.

The balance system \eqref{eq:balance-vector} is the one-dimensional
specialization of the general hierarchy introduced in
Section~\ref{sec:postulates-objectivity}, with
\[
 \bF^0=\bu,
 \qquad
 \bF^1=\bF.
\]
Similarly, \eqref{eq:supplementary-law} is the one-dimensional
specialization of the general supplementary law. Compatibility is
therefore governed by the main-field relations
\eqref{eq:general-main-field-relations} and the dual potentials
\eqref{eq:general-dual-potentials}. These general relations are not
repeated here; in the next section they are applied directly to the
elastic hierarchy.

Throughout the paper, scalar quantities are written in ordinary type,
whereas vectors and matrices are written in bold type. In particular,
\(\Psi\) is scalar, \(\bpsi\) and \(\be\) are column vectors, \(\bLambda\) is a row vector,
\(\bu'\) and its subvectors are row covectors, and \(\bJ,\bM\) are
matrices.

\section{Compatibility of the elastic hierarchy with the supplementary energy law}
\label{sec:supplementary-compatibility}

We now exploit the compatibility of the balance hierarchy with the
supplementary energy law. This determines the complete main field and
imposes precise restrictions on the total stress, the internal fluxes
and the productions. Introduce
\begin{equation}
 \bpsi(F,P,Q)
 =
 \begin{pmatrix}
  \psi_1(F,P,Q)\\
  \psi_2(F,P,Q)
 \end{pmatrix},
 \qquad
 \bJ
 =
 \frac{\partial\bpsi}{\partial(P,Q)},
 \label{eq:psi-J}
\end{equation}
and assume
\begin{equation}
 \det\bJ\neq0.
 \label{eq:J-nonzero}
\end{equation}
We also set
\begin{equation}
 \be
 =
 \begin{pmatrix}
  \Psi_P\\
  \Psi_Q
 \end{pmatrix},
 \qquad
 \bpsi_F
 =
 \begin{pmatrix}
  (\psi_1)_F\\
  (\psi_2)_F
 \end{pmatrix}.
 \label{eq:e-psiF}
\end{equation}

Before exploiting the supplementary law, we write the main field as
the row covector
\begin{equation}
 \bu'
 =
 \bigl(\zeta,\lambda,\Lambda_1,\Lambda_2\bigr),
 \label{eq:general-main-field-components}
\end{equation}
where
\[
 \bLambda
 =
 \bigl(\Lambda_1,\Lambda_2\bigr)
\]
is a row vector and all four components are initially unknown. The
components \(\zeta\) and \(\lambda\) are conjugate, respectively, to
the momentum density \(\rho v\) and the deformation density \(F\),
while \(\Lambda_1\) and \(\Lambda_2\) are conjugate to the two
additional balance densities. The symbol \(\zeta\) is used here to
avoid a conflict with the wave number and travelling-wave coordinate
denoted below by \(\xi\).

\begin{theorem}[Supplementary-law compatibility for the elastic hierarchy]
\label{thm:supplementary-compatibility}
Consider the one-dimensional hierarchy
\eqref{eq:kinematic}, \eqref{eq:momentum} and
\eqref{eq:additional-laws}, with supplementary energy density
\eqref{eq:general-energy} and energy flux
\eqref{eq:general-energy-flux}. Assume \eqref{eq:J-nonzero}, and assume
that \(\Omega_1,\Omega_2\) and \(k\) do not depend explicitly on the
velocity \(v\).

Then the supplementary law is compatible with the balance hierarchy
in the Ruggeri--Strumia sense if and only if the complete main field is
\begin{equation}
 \boxed{
 \bu'
 =
 \bigl(v,\cT,\Lambda_1,\Lambda_2\bigr).
 }
 \label{eq:complete-main-field}
\end{equation}
Here the internal part of the main field is the row vector
\begin{equation}
 \bLambda
 =
 \bigl(\Lambda_1,\Lambda_2\bigr)
 =
 \be^T\bJ^{-1}
 \label{eq:Lambda}
\end{equation}
and the total stress is
\begin{equation}
 \cT
 =
 \Psi_F-\bLambda\bpsi_F.
 \label{eq:stress-representation}
\end{equation}
Equivalently, all four components in
\eqref{eq:general-main-field-components} are determined by
\begin{equation}
 \zeta=v,
 \qquad
 \lambda=\cT,
 \qquad
 \bLambda=\be^T\bJ^{-1}.
 \label{eq:all-main-field-components}
\end{equation}
The additional energy flux and the two internal balance fluxes satisfy
the exactness condition
\begin{equation}
 \dd k
 =
 \Lambda_1\,\dd\Omega_1
 +
 \Lambda_2\,\dd\Omega_2.
 \label{eq:flux-exactness}
\end{equation}
Finally, the residual production is
\begin{equation}
 \Sigma
 =
 \Lambda_1\cP_1+\Lambda_2\cP_2
 \leq0.
 \label{eq:h-Sigma}
\end{equation}
\end{theorem}

Thus, compatibility with the supplementary energy law determines not only the internal
multipliers \(\Lambda_1,\Lambda_2\), but the complete main field. The
velocity \(v\) is the component conjugate to the momentum density,
whereas the total stress \(\cT\) is the component conjugate to the
deformation density. Once \(\Psi\) and the balance densities
\(\psi_1,\psi_2\) have been assigned, equations \eqref{eq:Lambda} and
\eqref{eq:stress-representation} determine all four components of
\(\bu'\). The internal fluxes and the additional energy flux must then
satisfy \eqref{eq:flux-exactness}, while the productions are restricted
by \eqref{eq:h-Sigma}. The proof is given in Appendix~A.

\begin{corollary}[Additive energy]
\label{cor:additive}
If
\begin{equation}
 \Psi(F,P,Q)
 =
 W(F)+E_P(P)+E_Q(Q),
 \label{eq:additive-energy}
\end{equation}
then
\begin{equation}
 \bLambda
 =
 \bigl(E_P'(P),E_Q'(Q)\bigr)\bJ^{-1},
 \label{eq:additive-multipliers}
\end{equation}
and
\begin{equation}
 \cT
 =
 W'(F)-\bLambda\bpsi_F.
 \label{eq:additive-stress}
\end{equation}
\end{corollary}

The general symmetrization and convexity theory has already been
established in Section~\ref{sec:postulates-objectivity} and is not
repeated here. Once the constitutive functions are specified, it
remains only to evaluate the intrinsic quadratic form
\(\dd\bu'\,\dd\bF^0\). This calculation will be carried out
explicitly for the canonical hierarchy in the next section.

\section{Canonical two-field RET hierarchy}

We now specialize the general compatibility theorem. The guiding
principle is to preserve, as the first link of the hierarchy, the
balance-density structure already used in RET viscoelasticity
\cite{RuggeriVisco2024}. We therefore identify the first internal
primitive variable with a
generalized stress \(\sigma\), while the second variable \(Q\)
represents the next internal level.

We choose
\begin{equation}
 \psi_1=Z(\sigma)-F,
 \qquad
 \psi_2=Y(Q),
 \qquad
 Z'(\sigma)Y'(Q)\neq0.
 \label{eq:canonical-densities}
\end{equation}
The first relation is equivalent, through \(F_t=v_z\), to the
viscoelastic-type balance
\[
 \partial_tZ(\sigma)-v_z+\partial_z\Omega_1=\cP_1.
\]
Thus the generalized-stress structure is built into the first
additional balance density, not introduced afterwards through a
gradient constitutive equation.

For the canonical realization we take the additive supplementary
stored energy
\begin{equation}
 \Psi(F,\sigma,Q)
 =
 W(F)+E_\sigma(\sigma)+E_Q(Q).
 \label{eq:canonical-energy}
\end{equation}
Then
\begin{equation}
 \bJ=
 \begin{pmatrix}
  Z'(\sigma)&0\\
  0&Y'(Q)
 \end{pmatrix},
 \qquad
 \bpsi_F=
 \begin{pmatrix}
  -1\\0
 \end{pmatrix}.
 \label{eq:canonical-J}
\end{equation}

The compatibility theorem gives
\begin{equation}
 \Lambda_1=
 \frac{E_\sigma'(\sigma)}{Z'(\sigma)},
 \qquad
 \Lambda_2=
 \frac{E_Q'(Q)}{Y'(Q)}.
 \label{eq:canonical-main-pre}
\end{equation}
To retain \(\sigma\) itself as the main field conjugate to the first
additional balance law, and hence as the generalized-stress
contribution to the total stress, we require
\begin{equation}
 E_\sigma'(\sigma)
 =
 \sigma Z'(\sigma).
 \label{eq:viscoelastic-compatibility}
\end{equation}
This is precisely the energy--balance-density relation appearing in
the one-field RET theory of viscoelasticity.

We define the second main field by
\begin{equation}
 \chi(Q)
 =
 \frac{E_Q'(Q)}{Y'(Q)}.
 \label{eq:chi-def}
\end{equation}
Hence
\begin{equation}
 \bLambda
 =
 \bigl(\sigma,\chi\bigr).
 \label{eq:canonical-multipliers}
\end{equation}
Since \(\bpsi_F=(-1,0)^T\), the stress representation becomes
\begin{equation}
 \cT=T(F)+\sigma,
 \qquad
 T(F)=W'(F).
 \label{eq:canonical-stress}
\end{equation}
Thus the first internal field has exactly the mathematical role of a
generalized stress, while the second field supplies the next
moment-like level.

\subsection{Convexity of the canonical hierarchy}

For the canonical choice
\eqref{eq:canonical-densities}--\eqref{eq:chi-def}, the main field is
\begin{equation}
 \bu'
 =
 \bigl(v,T(F)+\sigma,\sigma,\chi(Q)\bigr),
 \label{eq:canonical-main-field}
\end{equation}
whereas
\begin{equation}
 \bF^0=\bu
 =
 \bigl(
 \rho v,\,
 F,\,
 Z(\sigma)-F,\,
 Y(Q)
 \bigr)^T.
 \label{eq:canonical-F0}
\end{equation}
Therefore
\begin{align}
 \dd\bu'\,\dd\bF^0
 &=
 \rho(\dd v)^2
 +\dd F\,\dd[T(F)+\sigma]
 +\dd[Z(\sigma)-F]\,\dd\sigma
 +\dd Y(Q)\,\dd\chi(Q)
 \notag\\
 &=
 \rho(\dd v)^2
 +W''(F)(\dd F)^2
 +Z'(\sigma)(\dd\sigma)^2
 +Y'(Q)\chi'(Q)(\dd Q)^2.
 \label{eq:canonical-convexity-form}
\end{align}
The mixed terms \(\dd F\,\dd\sigma\) cancel identically. Hence the
Ruggeri--Strumia convexity condition is satisfied provided
\begin{equation}
 \rho>0,\qquad
 W''(F)>0,\qquad
 Z'(\sigma)>0,\qquad
 Y'(Q)\chi'(Q)>0.
 \label{eq:canonical-convexity-conditions}
\end{equation}
Since \(Y'(Q)\neq0\), the last condition is equivalently
\begin{equation}
 \frac{\dd\chi}{\dd Y}
 =
 \frac{\chi'(Q)}{Y'(Q)}>0.
 \label{eq:chi-Y-convexity}
\end{equation}

For the one-field principal subsystem obtained by fixing
\(\chi=0\), the condition reduces to
\begin{equation}
 \dd\bu'_{\rm red}\,\dd\bF^0_{\rm red}
 =
 \rho(\dd v)^2
 +W''(F)(\dd F)^2
 +Z'(\sigma)(\dd\sigma)^2>0,
 \label{eq:onefield-convexity}
\end{equation}
which is precisely the corresponding generalized-stress RET
convexity condition.

\subsection{Additional energy flux and internal balance fluxes}

The compatibility condition for the additional energy flux reads
\begin{equation}
 \dd k
 =
 \sigma\,\dd\Omega_1
 +
 \chi\,\dd\Omega_2.
 \label{eq:canonical-flux-compatibility}
\end{equation}
The simplest non-trivial bilinear choice is
\begin{equation}
 k=-a\sigma\chi,
 \qquad a\neq0.
 \label{eq:k-nearest}
\end{equation}
It is realized by
\begin{equation}
 \Omega_1=-a\chi,
 \qquad
 \Omega_2=-a\sigma,
 \label{eq:nearest-fluxes}
\end{equation}
since
\[
 \sigma\,\dd(-a\chi)
 +
 \chi\,\dd(-a\sigma)
 =
 -a\,\dd(\sigma\chi)
 =
 \dd k.
\]

The complete canonical hierarchy is therefore
\begin{subequations}
\label{eq:canonical-system}
\begin{align}
 &\rho v_t-\partial_z[T(F)+\sigma]=0,
 \label{eq:canonical-momentum}\\
 &F_t-v_z=0,
 \label{eq:canonical-kinematic}\\
 &\partial_t[Z(\sigma)-F]
 -a\,\partial_z\chi
 =\cP_1,
 \label{eq:canonical-sigma}\\
 &\partial_tY(Q)
 -a\,\partial_z\sigma
 =\cP_2.
 \label{eq:canonical-Q}
\end{align}
\end{subequations}
Here the constitutive functions entering the canonical hierarchy are
related by
\[
 T(F)=W'(F),
 \qquad
 E_\sigma'(\sigma)=\sigma Z'(\sigma),
 \qquad
 E_Q'(Q)=\chi(Q)Y'(Q),
\]
while strict convexity of the supplementary energy requires, on the
admissible state domain,
\[
 \rho>0,
 \qquad
 W''(F)>0,
 \qquad
 Z'(\sigma)>0,
 \qquad
 Y'(Q)\chi'(Q)>0.
\]
Its supplementary law is
\begin{align}
 &\partial_t\left[
 \frac{\rho}{2}v^2
 +W(F)+E_\sigma(\sigma)+E_Q(Q)
 \right]
 \notag\\
 &\qquad
 +\partial_z\left[
 -v\bigl(T(F)+\sigma\bigr)
 -a\sigma\chi
 \right]
 =
 \sigma\cP_1+\chi\cP_2
 \leq0.
 \label{eq:canonical-energy-law}
\end{align}
The term \(k=-a\sigma\chi\) is the local additional energy flux
associated with reversible transfer between the generalized-stress
mode and the higher internal mode.

\subsection{Production in main-field variables}

The production can be decomposed into a dissipative symmetric part
and an entropy-neutral antisymmetric part:
\begin{equation}
 (\cP_1,\cP_2)^T
 =
 -\bM\,(\sigma,\chi)^T
 +
 \kappa\,(-\chi,\sigma)^T,
 \qquad
 \bM=\bM^T\geq0.
 \label{eq:production-decomposition}
\end{equation}
Then
\begin{equation}
 \Sigma
 =
 -(\sigma,\chi)\,\bM\,(\sigma,\chi)^T
 \leq0.
 \label{eq:production-entropy}
\end{equation}
The antisymmetric term describes reversible exchange between
successive levels of the hierarchy.

\subsection{Principal subsystems and nested elastic reductions}
\label{subsec:principal-subsystems}

We first recall the definition introduced by Boillat and Ruggeri
\cite{BoillatRuggeri1997}. Split the row main field into two row
subvectors,
\begin{equation}
 \bu'=\bigl(\bm v',\bm w'\bigr),
 \label{eq:main-field-splitting}
\end{equation}
and split the corresponding column fluxes and production vector into
the conforming column blocks
\[
 \bF^\alpha=
 \begin{pmatrix}
  \bF_v^\alpha\\
  \bF_w^\alpha
 \end{pmatrix},
 \qquad
 \bP=
 \begin{pmatrix}
  \bP_v\\
  \bP_w
 \end{pmatrix}.
\]
Derivatives of scalar potentials with respect to the row subvector
\(\bm v'\) are understood as column gradients.
A \emph{principal subsystem} is obtained by assigning a constant value
\begin{equation}
 \bm w'=\bm w'_{*}
 \label{eq:frozen-main-field}
\end{equation}
to one or more components of the main field and omitting the balance
equations conjugate to those frozen components. The remaining equations
are therefore
\begin{equation}
 \partial_t
 \left[
  \frac{\partial h'^0}{\partial\bm v'}
  (\bm v',\bm w'_{*})
 \right]
 +
 \partial_z
 \left[
  \frac{\partial h'^1}{\partial\bm v'}
  (\bm v',\bm w'_{*})
 \right]
 =
 \bP_v(\bm v',\bm w'_{*}).
 \label{eq:general-principal-subsystem}
\end{equation}
The terminology \emph{principal} reflects the fact that the symmetric
coefficient matrices of the reduced system are principal submatrices
of the Hessian matrices of the full system.

Boillat and Ruggeri proved two results which are essential here. First,
every principal subsystem of a symmetric-hyperbolic balance system with
a convex supplementary density is itself symmetric hyperbolic and
inherits a convex supplementary law, called the \emph{subentropy law}.
In the notation above, its density and flux may be written as
\begin{equation}
 \bar h^\alpha
 =
 \left.
 \left(
  h^\alpha-\bm w'_{*}\,\bF_w^\alpha
 \right)
 \right|_{\bm w'=\bm w'_{*}},
 \qquad
 \alpha=0,1,
 \label{eq:subentropy-fluxes}
\end{equation}
with reduced production
\begin{equation}
 \bar\Sigma
 =
 \left.
 \left(
  \Sigma-\bm w'_{*}\,\bP_w
 \right)
 \right|_{\bm w'=\bm w'_{*}}.
 \label{eq:subentropy-production}
\end{equation}
Second, the characteristic speeds satisfy the subcharacteristic
conditions. If \(\lambda_{\min}\) and \(\lambda_{\max}\) are the
extreme characteristic speeds of the full system evaluated on the
constrained state and \(\bar\lambda_{\min}\),
\(\bar\lambda_{\max}\) are those of the principal subsystem, then
\begin{equation}
 \lambda_{\min}
 \leq
 \bar\lambda_{\min},
 \qquad
 \bar\lambda_{\max}
 \leq
 \lambda_{\max}.
 \label{eq:principal-subcharacteristic}
\end{equation}
Thus a principal reduction cannot create a faster right-going wave or
a faster left-going wave than those already present in the parent
system \cite{BoillatRuggeri1997}.

For the canonical hierarchy, the main field is given by
\eqref{eq:canonical-main-field}. We first freeze its last component,
\begin{equation}
 \chi=\chi_{\rm eq}=0,
 \label{eq:chi-principal}
\end{equation}
and omit the second additional balance law. Under the equilibrium
normalization \(Q=Q_{\rm eq}\), one has
\(\Omega_1=-a\chi=0\), and the remaining system is
\begin{subequations}
\label{eq:one-field-principal}
\begin{align}
 &\rho v_t-\partial_z[T(F)+\sigma]=0,
 \label{eq:onefield-momentum}\\
 &F_t-v_z=0,
 \label{eq:onefield-kinematic}\\
 &\partial_t[Z(\sigma)-F]
 =\cP_{1,\rm red}(F,\sigma).
 \label{eq:onefield-stress}
\end{align}
\end{subequations}
Using \(F_t=v_z\), the last equation is equivalently
\begin{equation}
 \partial_t Z(\sigma)-v_z
 =
 \cP_{1,\rm red}(F,\sigma).
 \label{eq:onefield-visco-form}
\end{equation}
This is precisely the balance-density structure of the one-field RET
theory of viscoelasticity. In the present setting \(\sigma\) is best
regarded more generally as a dynamically evolving stress contribution;
no viscous interpretation is required for the two-field parent model.

A second principal reduction is obtained by freezing the remaining
internal main-field component at its equilibrium value,
\begin{equation}
 \sigma=\sigma_{\rm eq}=0,
 \label{eq:sigma-principal}
\end{equation}
and omitting the generalized-stress balance. The resulting system is
one-dimensional nonlinear elasticity,
\begin{equation}
 \rho v_t-[T(F)]_z=0,
 \qquad
 F_t-v_z=0,
 \qquad
 T(F)=W'(F).
 \label{eq:pure-elasticity}
\end{equation}
It possesses the supplementary conservation law
\begin{equation}
 \partial_t
 \left[
  \frac{\rho}{2}v^2+W(F)
 \right]
 +
 \partial_z[-vT(F)]
 =0.
 \label{eq:pure-elastic-energy}
\end{equation}
For the balance-density vector
\begin{equation}
 \bu_{\rm el}
 =
 \begin{pmatrix}
  \rho v\\
  F
 \end{pmatrix},
 \label{eq:elastic-density-vector}
\end{equation}
the induced main field is
\begin{equation}
 \bu'_{\rm el}
 =
 \bigl(v,T(F)\bigr).
 \label{eq:elastic-main-field}
\end{equation}
Indeed,
\(\dd h^0_{\rm el}=\bu'_{\rm el}\,\dd\bu_{\rm el}\), and the
convexity form reduces to
\begin{equation}
 \dd\bu'_{\rm el}\,\dd\bu_{\rm el}
 =
 \rho(\dd v)^2+T'(F)(\dd F)^2.
 \label{eq:elastic-convexity}
\end{equation}
Hence the elastic principal subsystem is symmetric hyperbolic whenever
\begin{equation}
 T'(F)=W''(F)>0.
 \label{eq:elastic-hyperbolicity}
\end{equation}
This is the one-dimensional specialization of the symmetric
formulation of nonlinear mechanics obtained by Boillat and Ruggeri
\cite{BoillatRuggeri1980}.

The canonical hierarchy therefore has the exact nested structure
\begin{equation}
 \text{nonlinear elasticity}
 \subset
 \text{one-field generalized-stress RET}
 \subset
 \text{two-field dispersive RET}.
 \label{eq:principal-hierarchy}
\end{equation}
This nesting is one reason for preferring
\(\psi_1=Z(\sigma)-F\): the viscoelastic model and nonlinear elasticity
appear as genuine principal subsystems rather than merely as formal
analogies. By the Boillat--Ruggeri theorem, each reduction inherits a
convex supplementary law and satisfies the corresponding
subcharacteristic inequalities.

\section{Linear reversible theory}
\label{sec:linear-theory}

We now specialize the canonical hierarchy to a quadratic higher-mode
energy and study the resulting linear theory before introducing any
nonlinear generalization.  This is the level at which the signs of the
internal inertias, the dispersion relation, and the realizable
Love--Rosenau coefficient domain can be determined without ambiguity.

\subsection{Quadratic reversible specialization}

We now choose
\begin{equation}
 Z(\sigma)=\tau_\sigma\sigma,
 \qquad
 E_\sigma(\sigma)=\frac{\tau_\sigma}{2}\sigma^2,
 \label{eq:quadratic-sigma}
\end{equation}
and
\begin{equation}
 Y(Q)=\tau_QQ,
 \qquad
 E_Q(Q)=\frac{\tau_Q}{2}Q^2,
 \label{eq:quadratic-Q}
\end{equation}
with
\begin{equation}
 \tau_\sigma>0,\qquad \tau_Q>0.
\end{equation}
Then \eqref{eq:viscoelastic-compatibility} holds identically and
\begin{equation}
 \chi=Q.
\end{equation}

For the purely reversible antisymmetric production take
\begin{equation}
 \cP_1=-\kappa Q,
 \qquad
 \cP_2=\kappa\sigma,
 \qquad
 \kappa\neq0.
 \label{eq:antisymmetric-production}
\end{equation}
The complete reversible hierarchy becomes
\begin{subequations}
\label{eq:reversible-system}
\begin{align}
 &\rho v_t-\partial_z[T(F)+\sigma]=0,
 \label{eq:rev-momentum}\\
 &F_t-v_z=0,
 \label{eq:rev-kinematic}\\
 &\tau_\sigma\sigma_t-F_t-aQ_z=-\kappa Q,
 \label{eq:rev-sigma}\\
 &\tau_QQ_t-a\sigma_z=\kappa\sigma.
 \label{eq:rev-Q}
\end{align}
\end{subequations}
It satisfies the exact conservation law
\begin{align}
 &\partial_t\left[
 \frac{\rho}{2}v^2+W(F)
 +\frac{\tau_\sigma}{2}\sigma^2
 +\frac{\tau_Q}{2}Q^2
 \right]
 \notag\\
 &\qquad
 +\partial_z\left[
 -v\bigl(T(F)+\sigma\bigr)-a\sigma Q
 \right]=0.
 \label{eq:reversible-energy}
\end{align}
In this quadratic case,
\begin{equation}
 Z'(\sigma)=\tau_\sigma,
 \qquad
 Y'(Q)=\tau_Q,
 \qquad
 \chi'(Q)=1,
\end{equation}
and therefore the Ruggeri--Strumia quadratic form
\eqref{eq:canonical-convexity-form} reduces to
\begin{equation}
 \dd\bu'\,\dd\bF^0
 =
 \rho(\dd v)^2
 +W''(F)(\dd F)^2
 +\tau_\sigma(\dd\sigma)^2
 +\tau_Q(\dd Q)^2.
 \label{eq:quadratic-RS-convexity}
\end{equation}
Consequently the full reversible hierarchy is symmetric hyperbolic
for
\begin{equation}
 W''(F)>0,\qquad
 \tau_\sigma>0,\qquad
 \tau_Q>0.
 \label{eq:quadratic-convexity-conditions}
\end{equation}
For linear elasticity \(W(F)=\mu F^2/2\), this becomes simply
\begin{equation}
 \dd\bu'\,\dd\bF^0
 =
 \rho(\dd v)^2
 +\mu(\dd F)^2
 +\tau_\sigma(\dd\sigma)^2
 +\tau_Q(\dd Q)^2>0.
 \label{eq:linear-RS-convexity}
\end{equation}

For the remainder of this section we specialize further to
\begin{equation}
 W(F)=\frac{\mu}{2}F^2,
 \qquad T(F)=\mu F,
 \qquad \mu>0.
 \label{eq:linear-elastic-specialization}
\end{equation}
Thus the full parent system is linear, while its four-field energy is
strictly convex for
\[
 \rho>0,\qquad \mu>0,\qquad
 \tau_\sigma>0,\qquad \tau_Q>0.
\]
In particular, the positivity of \(\tau_Q\) is a direct consequence of
Ruggeri--Strumia convexity and is not an additional sign convention.

\subsection{Dispersion relation of the four-field parent system}

For \(T(F)=\mu F\), plane waves
\(e^{i(\xi z-\omega t)}\) in \eqref{eq:reversible-system} satisfy
\begin{align}
 &\rho\tau_\sigma\tau_Q\omega^4
 -
 \Bigl[
 \kappa^2\rho
 +
 \xi^2\bigl(
 a^2\rho+\mu\tau_\sigma\tau_Q+\tau_Q
 \bigr)
 \Bigr]\omega^2
 \notag\\
 &\qquad
 +
 \mu\kappa^2\xi^2
 +
 a^2\mu\xi^4
 =0.
 \label{eq:quartic-dispersion}
\end{align}
For fixed \(\xi\), the two roots in \(\omega^2\) separate as
\(\tau_\sigma\to0\). The fast root satisfies
\begin{equation}
 \omega_{\rm fast}^2
 =
 \frac{\kappa^2\rho+\xi^2(\rho a^2+\tau_Q)}
 {\rho\tau_\sigma\tau_Q}
 +O(1),
 \qquad \tau_\sigma\to0,
 \label{eq:fast-frequency-asymptotics}
\end{equation}
whereas the other root remains finite and tends to
\begin{equation}
 \left[
 \rho+\frac{\rho a^2+\tau_Q}{\kappa^2}\xi^2
 \right]\omega^2
 =
 \mu\xi^2+\frac{\mu a^2}{\kappa^2}\xi^4.
 \label{eq:exact-LR-dispersion}
\end{equation}
Thus \(\tau_\sigma\) controls a genuinely fast inertial mode, while the
higher field retains the finite inertia \(\tau_Q\). This is a
separation of inertial time scales, not a dissipative relaxation.

For fixed \(\tau_\sigma>0\), the acoustic branch has
\begin{equation}
 \omega^2
 =
 \frac{\mu}{\rho}\xi^2
 -
 \frac{\mu\tau_Q}{\kappa^2\rho^2}\xi^4
 +O(\xi^6).
 \label{eq:acoustic-expansion}
\end{equation}

\subsection{Singular fast-stress limit}

Consider the distinguished limit
\begin{equation}
 \tau_\sigma\longrightarrow0,
 \label{eq:linear-fast-stress-limit}
\end{equation}
with \(\tau_Q,a,\kappa\) fixed.  The production remains antisymmetric,
so this is a reversible vanishing-inertia limit rather than a
dissipative Maxwellian iteration.

\begin{theorem}[Linear Love--Rosenau reduction]
\label{thm:linear-LR-reduction}
In the limit \(\tau_\sigma=0\), elimination of \(\sigma\) and \(Q\)
from the linear specialization of \eqref{eq:reversible-system} yields
\begin{equation}
 \rho u_{tt}-\mu u_{zz}
 +\frac{\mu a^2}{\kappa^2}u_{zzzz}
 -\frac{\rho a^2+\tau_Q}{\kappa^2}u_{zztt}=0.
 \label{eq:linear-LR-reduced}
\end{equation}
Hence
\begin{equation}
 \alpha=\frac{\mu a^2}{\kappa^2}>0,
 \qquad
 \beta=\frac{\rho a^2+\tau_Q}{\kappa^2}>0,
 \label{eq:linear-alpha-beta}
\end{equation}
and
\begin{equation}
 \beta-\frac{\rho}{\mu}\alpha
 =\frac{\tau_Q}{\kappa^2}>0.
 \label{eq:linear-alpha-beta-gap}
\end{equation}
\end{theorem}

\begin{proof}
Set
\[
 R:=\rho v_t-\mu F_z.
\]
The momentum equation gives \(\sigma_z=R\).  For
\(\tau_\sigma=0\), the generalized-stress equation and
\(F_t=v_z\) give
\[
 (\kappa-a\partial_z)Q=v_z.
\]
Differentiating the \(Q\)-equation with respect to \(z\), applying
\(\kappa-a\partial_z\), and using the two preceding identities yields
\[
 \kappa^2R-a^2R_{zz}-\tau_Qv_{tzz}=0.
\]
Substitution of \(F=u_z\) and \(v=u_t\) gives
\eqref{eq:linear-LR-reduced}.
\end{proof}

The high-frequency limiting phase speed and the long-wave sound speed
satisfy
\begin{equation}
 c_\infty^2=\frac{\alpha}{\beta}
 <\frac{\mu}{\rho}=c_0^2.
 \label{eq:linear-subcharacteristic}
\end{equation}
The strict inequality is again equivalent to \(\tau_Q>0\).

\section{Linear realizability of the Love--Rosenau model}

The quadratic reversible hierarchy yields
\[
 \alpha=\frac{\mu a^2}{\kappa^2},
 \qquad
 \beta=\frac{\rho a^2+\tau_Q}{\kappa^2},
\]
and therefore
\begin{equation}
 \beta-\frac{\rho}{\mu}\alpha
 =
 \frac{\tau_Q}{\kappa^2}>0.
 \label{eq:minimal-constraint}
\end{equation}
We now show that this restriction is not an artifact of the particular
choice \(k=-a\sigma Q\). At the linear level it follows from the
general structure of a convex reversible two-field RET realization
with one singular fast mode.

\subsection{General linear reversible two-field structure}

We retain the notation already used for the internal part of the main
field and write it as the row vector
\begin{equation}
 \bLambda=(\Lambda_1,\Lambda_2).
 \label{eq:general-linear-internal-main-field}
\end{equation}
After a nonsingular linear change of the internal main-field
components, a general normal form for the reversible two-field class
considered here may be written as
\begin{subequations}
\label{eq:general-linear-twofield}
\begin{align}
 &\rho v_t
 -
 \partial_z\!\left(
 \mu F+\bLambda\bc
 \right)=0,
 \label{eq:general-linear-momentum}\\
 &F_t-v_z=0,
 \label{eq:general-linear-kinematic}\\
 &\bLambda_t\bC
 -F_t\bc^T
 +\bLambda_z\bB
 =\bLambda\bA.
 \label{eq:general-linear-internal}
\end{align}
\end{subequations}
Here
\begin{equation}
 \bC=\bC^T>0,
 \qquad
 \bB=\bB^T,
 \qquad
 \bA=-\bA^T,
 \label{eq:general-linear-structure}
\end{equation}
where the positivity of \(\bC\) is the internal part of the
Ruggeri--Strumia convexity condition, the symmetry of \(\bB\) follows
from the existence of the spatial RET potential, and the
antisymmetry of \(\bA\) expresses entropy-neutral reversible
production.

We consider a distinguished singular limit in which exactly one
internal inertia vanishes. The limiting inertia matrix is symmetric,
positive semidefinite and of rank one. Hence a nonsingular linear
change of the two internal main-field components brings it, by
congruence, to the normal form
\begin{equation}
 \bC_0=
 \begin{pmatrix}
  0&0\\
  0&\tau
 \end{pmatrix},
 \qquad
 \tau>0.
 \label{eq:C0-normal}
\end{equation}
The symmetry of \(\bB\) and the antisymmetry of \(\bA\) are preserved
under the corresponding change of variables. Moreover, every real
\(2\times2\) antisymmetric matrix is a scalar multiple of the canonical
skew-symmetric matrix. Therefore, for a genuinely coupled reversible
production and a row internal main field, one may write
\begin{equation}
 \bA
 =
 \kappa
 \begin{pmatrix}
  0&1\\
  -1&0
 \end{pmatrix},
 \qquad
 \kappa\neq0.
 \label{eq:A-normal}
\end{equation}
Finally, we set
\begin{equation}
 \bB=
 \begin{pmatrix}
  b&a\\
  a&d
 \end{pmatrix},
 \qquad
 \bc=(\gamma,\delta)^T.
 \label{eq:Bc-general}
\end{equation}

\begin{theorem}[Two-field realizability of the Love--Rosenau model]
\label{thm:twofield-realizability}
Consider the singular linear reversible two-field RET system
\eqref{eq:general-linear-twofield} with
\eqref{eq:C0-normal}--\eqref{eq:Bc-general}. Assume that, for each
real wavenumber \(\xi\), its finite acoustic sector consists of the
two counter-propagating branches
\(\omega=\pm\omega_{\rm LR}(\xi)\), where
\begin{equation}
 (\rho+\beta\xi^2)\omega_{\rm LR}^2
 =
 \mu\xi^2+\alpha\xi^4,
 \qquad
 \alpha>0,\quad\beta>0.
 \label{eq:LR-dispersion-realizability}
\end{equation}
Then necessarily
\begin{equation}
 \beta-\frac{\rho}{\mu}\alpha
 =
 \frac{\gamma^2\tau}{\kappa^2}\geq0.
 \label{eq:general-twofield-constraint}
\end{equation}
For a genuinely coupled fast mode, \(\gamma\neq0\), one has the strict
inequality
\begin{equation}
 \beta>\frac{\rho}{\mu}\alpha.
 \label{eq:strict-twofield-constraint}
\end{equation}

Conversely, for every pair
\begin{equation}
 \alpha>0,
 \qquad
 \beta>\frac{\rho}{\mu}\alpha,
 \label{eq:admissible-pair}
\end{equation}
there exists a convex reversible two-field RET realization of the
above singular type whose reduced dispersion relation is exactly
\eqref{eq:LR-dispersion-realizability}.
\end{theorem}

\begin{proof}
Insert a plane wave
\begin{equation}
 (v,F,\bLambda)
 =
 (\widehat v,\widehat F,\widehat{\bLambda})
 e^{i(\xi z-\omega t)},
 \qquad
 \widehat{\bLambda}
 =
 (\widehat\Lambda_1,\widehat\Lambda_2).
\end{equation}
into \eqref{eq:general-linear-twofield}. In the singular limit
\(\bC=\bC_0\), direct elimination gives the characteristic polynomial
\begin{align}
0={}&
-b\rho\tau\,\omega^3\xi
\notag\\
&+
\left[
 \rho(bd-a^2)-\gamma^2\tau
\right]
\omega^2\xi^2
-\rho\kappa^2\omega^2
\notag\\
&+
\left[
 b\delta^2+b\mu\tau
 -2a\gamma\delta+d\gamma^2
\right]
\omega\xi^3
\notag\\
&+
\mu(a^2-bd)\xi^4
+\mu\kappa^2\xi^2.
\label{eq:general-characteristic-polynomial}
\end{align}

By assumption, both
\(\omega=\omega_{\rm LR}(\xi)\) and
\(\omega=-\omega_{\rm LR}(\xi)\) are roots of
\eqref{eq:general-characteristic-polynomial}. Subtracting the two
characteristic equations and dividing by \(2\omega_{\rm LR}\xi\) for
\(\xi\neq0\) gives
\begin{equation}
 -b\rho\tau\,\omega_{\rm LR}^2
 +
 \left(
 b\delta^2+b\mu\tau-2a\gamma\delta+d\gamma^2
 \right)\xi^2
 =0.
 \label{eq:odd-branch-identity}
\end{equation}
Using \eqref{eq:LR-dispersion-realizability}, multiplying by
\(\rho+\beta\xi^2\), and comparing the constant and \(\xi^2\)
coefficients yields
\begin{equation}
 b\delta^2+d\gamma^2-2a\gamma\delta=0,
 \qquad
 b\tau(\beta\mu-\rho\alpha)=0.
 \label{eq:odd-branch-conditions}
\end{equation}

Adding instead the two characteristic equations eliminates the odd
powers of \(\omega\). Substitution of
\eqref{eq:LR-dispersion-realizability} and comparison of the
coefficient of \(\xi^2\) gives
\begin{equation}
 \kappa^2(\beta\mu-\rho\alpha)
 =
 \gamma^2\mu\tau.
 \label{eq:even-branch-condition}
\end{equation}
Therefore
\begin{equation}
 \beta-\frac{\rho}{\mu}\alpha
 =
 \frac{\gamma^2\tau}{\kappa^2}\geq0,
 \label{eq:general-constraint-proof}
\end{equation}
which proves the necessary constraint. If the fast mode is genuinely
coupled, \(\gamma\neq0\), then
\(\beta\mu-\rho\alpha>0\); the second relation in
\eqref{eq:odd-branch-conditions} consequently forces
\begin{equation}
 b=0.
 \label{eq:b-zero}
\end{equation}
The first relation in \eqref{eq:odd-branch-conditions} then becomes
\begin{equation}
 d\gamma^2-2a\gamma\delta=0.
 \label{eq:odd-second-condition}
\end{equation}
The remaining coefficient identity gives
\(\alpha=\mu a^2/\kappa^2\), and hence
\begin{equation}
 \alpha=\frac{\mu a^2}{\kappa^2},
 \qquad
 \beta=\frac{\rho a^2+\gamma^2\tau}{\kappa^2}.
 \label{eq:general-alpha-beta}
\end{equation}
Equivalently, the characteristic polynomial reduces to
\begin{equation}
 -
 \left[
 \rho\kappa^2+
 (\rho a^2+\gamma^2\tau)\xi^2
 \right]\omega^2
 +
 \mu\kappa^2\xi^2
 +
 \mu a^2\xi^4
 =0.
 \label{eq:reduced-general-characteristic}
\end{equation}
Since \(\tau>0\) by convexity, the inequality in
\eqref{eq:general-constraint-proof} is strict whenever
\(\gamma\neq0\).

Conversely, let
\[
 \alpha>0,
 \qquad
 \beta>\frac{\rho}{\mu}\alpha.
\]
Choose any \(\kappa\neq0\), set
\begin{equation}
 \gamma=1,
 \qquad
 a=\kappa\sqrt{\frac{\alpha}{\mu}},
 \qquad
 \tau=
 \kappa^2
 \left(
 \beta-\frac{\rho}{\mu}\alpha
 \right)>0,
 \label{eq:converse-parameters}
\end{equation}
and take
\begin{equation}
 b=d=\delta=0.
 \label{eq:converse-simple-choice}
\end{equation}
Then
\eqref{eq:general-alpha-beta} gives precisely the prescribed
\((\alpha,\beta)\), while \(\tau>0\) preserves the internal
Ruggeri--Strumia convexity. This proves the converse.
\end{proof}

\begin{corollary}[Admissible coefficient domain]
\label{cor:admissible-domain}
A genuinely dispersive Love--Rosenau equation admits a convex
reversible two-field RET realization of the singular type considered
above if and only if
\begin{equation}
 \boxed{
 \alpha>0,
 \qquad
 \beta>\frac{\rho}{\mu}\alpha.
 }
 \label{eq:admissible-domain}
\end{equation}
The canonical quadratic hierarchy realizes the entire admissible
domain.
\end{corollary}

\begin{corollary}[Subcharacteristic consequence of the principal subsystem]
\label{cor:subcharacteristic}
The elastic principal subsystem and the finite acoustic branch of every
genuinely coupled admissible two-field realization satisfy
\begin{equation}
 c_\infty^2
 =
 \frac{\alpha}{\beta}
 <
 \frac{\mu}{\rho}
 =
 c_0^2.
 \label{eq:subcharacteristic-speeds}
\end{equation}
Thus the high-frequency limiting phase speed of the reduced
Love--Rosenau model is strictly smaller than the characteristic speed
of its elastic principal subsystem.
\end{corollary}

\begin{proof}
The equilibrium elastic equations are obtained by freezing both
internal components of the main field and omitting the corresponding
balance laws; hence they form a principal subsystem in the sense of
Boillat and Ruggeri, as described in
Section~\ref{subsec:principal-subsystems}. Their characteristic speeds
are \(\pm c_0\), where \(c_0^2=\mu/\rho\). The finite acoustic branch
of the singular parent system tends, for large wave number, to
\(\pm c_\infty\), where \(c_\infty^2=\alpha/\beta\). The
Boillat--Ruggeri subcharacteristic inequalities therefore give
\(c_\infty\leq c_0\). For a genuinely coupled fast mode,
\(\gamma\neq0\), relation
\eqref{eq:general-twofield-constraint} makes the inequality strict.
\end{proof}

\begin{remark}[Degenerate boundary]
If
\begin{equation}
 \beta=\frac{\rho}{\mu}\alpha,
\end{equation}
then
\begin{equation}
 \frac{\mu\xi^2+\alpha\xi^4}
      {\rho+\beta\xi^2}
 =
 \frac{\mu}{\rho}\xi^2,
\end{equation}
so the apparent higher-order terms cancel from the dispersion
relation. The boundary of the admissible domain therefore corresponds
to a non-dispersive degeneration.
\end{remark}

\begin{remark}[Meaning of the result]
The inequality
\(
\beta>\rho\alpha/\mu
\)
is not a peculiarity of the nearest-neighbour choice
\(k=-a\sigma Q\). At the linear level it follows from the combination
of two internal fields, Ruggeri--Strumia convexity, reversible
entropy-neutral production, and a one-mode singular reduction leading
exactly to the Love--Rosenau dispersion law. The canonical hierarchy
constructed in this paper is therefore not merely one example: it is
a normal-form realization of the whole admissible coefficient domain.
\end{remark}

A remaining nonlinear question concerns the precise class of reduced
higher-gradient equations obtainable when the elastic stress is
nonlinear. In particular, our canonical model generates
\[
 \frac{a^2}{\kappa^2}[T(u_z)]_{zzz},
\]
rather than an independently assigned fourth-order constitutive term.
A second issue is the asymptotic relation between the local
microstructural energy flux \(k\) and the interstitial-working flux of
the preceding analysis.

\section{Nonlinear extension}
\label{sec:nonlinear-extension}

The preceding sections establish the complete linear theory.  We now
return to the nonlinear freedom already contained in the canonical
RET representation.  The elastic energy \(W(F)\), the generalized
stress density \(Z(\sigma)\), and the higher-field pair
\(Y(Q),E_Q(Q)\) may all be nonlinear, subject to
\begin{equation}
 E_\sigma'(\sigma)=\sigma Z'(\sigma),
 \qquad
 \chi(Q)=\frac{E_Q'(Q)}{Y'(Q)},
 \label{eq:nonlinear-general-relations}
\end{equation}
and to the convexity requirements
\begin{equation}
 W''(F)>0,
 \qquad Z'(\sigma)>0,
 \qquad Y'(Q)\chi'(Q)>0.
 \label{eq:nonlinear-general-convexity}
\end{equation}
Thus a general dependence of the energy on the higher internal field
is fully compatible with RET, but it must be coordinated with the
balance density \(Y(Q)\); changing the energy alone is not sufficient.

\subsection{A non-quadratic higher-field energy}
\label{subsec:nonquadratic-Q}

The quadratic choice \eqref{eq:quadratic-Q} is not forced by the
compatibility theorem.  A particularly simple non-quadratic
extension preserves the main-field component \(\chi=Q\), and hence
also the flux \(k=-a\sigma Q\), by taking
\begin{equation}
 Y(Q)=\tau_2Q+\tau_4Q^3,
 \qquad
 E_Q(Q)=\frac{\tau_2}{2}Q^2+\frac{3\tau_4}{4}Q^4,
 \label{eq:nonquadratic-Q-choice}
\end{equation}
where
\begin{equation}
 \tau_2>0,
 \qquad
 \tau_4\geq0.
 \label{eq:nonquadratic-Q-parameters}
\end{equation}
Indeed,
\begin{equation}
 E_Q'(Q)=QY'(Q),
 \qquad
 \chi(Q)=\frac{E_Q'(Q)}{Y'(Q)}=Q.
 \label{eq:nonquadratic-Q-compatibility}
\end{equation}
Therefore the reversible system becomes
\begin{subequations}
\label{eq:nonquadratic-reversible-system}
\begin{align}
 &\rho v_t-\partial_z[T(F)+\sigma]=0,
 \\
 &F_t-v_z=0,
 \\
 &\tau_\sigma\sigma_t-F_t-aQ_z=-\kappa Q,
 \\
 &\partial_t(\tau_2Q+\tau_4Q^3)-a\sigma_z=\kappa\sigma.
 \label{eq:nonquadratic-Q-balance}
\end{align}
\end{subequations}
Equivalently, the last equation reads
\begin{equation}
 (\tau_2+3\tau_4Q^2)Q_t-a\sigma_z=\kappa\sigma.
 \label{eq:nonquadratic-Q-evolution}
\end{equation}
The exact energy conservation law is
\begin{align}
 &\partial_t\left[
 \frac{\rho}{2}v^2+W(F)
 +\frac{\tau_\sigma}{2}\sigma^2
 +\frac{\tau_2}{2}Q^2+\frac{3\tau_4}{4}Q^4
 \right]
 \notag\\
 &\qquad
 +\partial_z\left[-v\bigl(T(F)+\sigma\bigr)-a\sigma Q\right]=0.
 \label{eq:nonquadratic-energy-law}
\end{align}
Moreover,
\begin{equation}
 \dd\bu'\,\dd\bF^0
 =\rho(\dd v)^2+W''(F)(\dd F)^2
 +\tau_\sigma(\dd\sigma)^2
 +(\tau_2+3\tau_4Q^2)(\dd Q)^2.
 \label{eq:nonquadratic-RS-convexity}
\end{equation}
Thus the hierarchy remains strictly symmetric hyperbolic for
\begin{equation}
 W''(F)>0,
 \qquad
 \tau_\sigma>0,
 \qquad
 \tau_2>0,
 \qquad
 \tau_4\geq0.
 \label{eq:nonquadratic-convexity-conditions}
\end{equation}

This example shows that the quadratic higher-mode energy is not forced
by the compatibility theorem.  It is included only to illustrate how
nonlinear balance densities and energies are introduced without
changing the main-field component \(\chi=Q\).  No travelling-wave or
compact-support conclusion is drawn here.

\subsection{Exact nonlinear Love--Rosenau reduction}

For the exact Love--Rosenau reduction we now return to the quadratic
higher-field choice \eqref{eq:quadratic-Q}, while retaining a general
nonlinear elastic stress \(T(F)\). Thus the non-quadratic example of
the preceding subsection establishes the wider admissible RET
architecture, but is not used in the exact reduction below.

The decisive reduction is obtained when the generalized-stress mode
becomes fast:
\begin{equation}
 \tau_\sigma\longrightarrow0,
 \label{eq:tausigma-limit}
\end{equation}
while \(\tau_Q,a,\kappa\) remain finite. The higher field \(Q\)
retains its inertia.

\begin{theorem}[Exact Love--Rosenau reduction]
\label{thm:exact-LR-reduction}
In the singular limit \(\tau_\sigma=0\), elimination of
\(\sigma\) and \(Q\) from \eqref{eq:reversible-system} gives the exact
nonlinear displacement equation
\begin{equation}
 \boxed{
 \rho u_{tt}
 -[T(u_z)]_z
 +\frac{a^2}{\kappa^2}[T(u_z)]_{zzz}
 -\frac{\rho a^2+\tau_Q}{\kappa^2}u_{zztt}
 =0.
 }
 \label{eq:nonlinear-LR-limit}
\end{equation}
\end{theorem}

\begin{proof}
Define
\begin{equation}
 R:=\rho v_t-[T(F)]_z.
 \label{eq:R-def}
\end{equation}
From the momentum equation,
\begin{equation}
 \sigma_z=R.
 \label{eq:sigmaz-R}
\end{equation}
For \(\tau_\sigma=0\), equation \eqref{eq:rev-sigma}, together with
\(F_t=v_z\), becomes
\begin{equation}
 (\kappa-a\partial_z)Q=v_z.
 \label{eq:Q-operator}
\end{equation}
Differentiating \eqref{eq:rev-Q} with respect to \(z\) and using
\eqref{eq:sigmaz-R} gives
\begin{equation}
 \tau_Q Q_{tz}-aR_z=\kappa R.
 \label{eq:Qt-R}
\end{equation}
Apply \(\kappa-a\partial_z\) to \eqref{eq:Qt-R}. From
\eqref{eq:Q-operator},
\[
 (\kappa-a\partial_z)Q_{tz}=v_{tzz}.
\]
Therefore
\[
 \tau_Qv_{tzz}
 -a(\kappa-a\partial_z)R_z
 =
 \kappa(\kappa-a\partial_z)R.
\]
The terms proportional to \(R_z\) cancel, and we obtain
\begin{equation}
 \kappa^2R-a^2R_{zz}-\tau_Qv_{tzz}=0.
 \label{eq:R-reduced}
\end{equation}
Since \(F=u_z\) and \(v=u_t\), substitution of
\eqref{eq:R-def} gives \eqref{eq:nonlinear-LR-limit}.
\end{proof}

Equation \eqref{eq:nonlinear-LR-limit} can be written as
\begin{equation}
 \left(1-\ell^2\partial_{zz}\right)
 \left\{
 \rho u_{tt}-[T(u_z)]_z
 \right\}
 -\delta\,u_{zztt}=0,
 \label{eq:compact-LR}
\end{equation}
with
\begin{equation}
 \ell^2=\frac{a^2}{\kappa^2},
 \qquad
 \delta=\frac{\tau_Q}{\kappa^2}>0.
 \label{eq:ell-delta}
\end{equation}

\begin{remark}[Hyperbolic parent system, dispersive reduction and dissipative pulses]
\label{rem:hyperbolic-dispersive-pulse}
The limit \(\tau_\sigma\to0\) is singular. For every positive
\(\tau_\sigma\) satisfying the convexity conditions, the parent RET
system remains first-order and symmetric hyperbolic. This does not
exclude dispersion: the internal balance fields and the reversible
production introduce an intrinsic frequency scale, so that the linear
branches of the hyperbolic balance system have frequency-dependent
phase velocities. Eliminating the fast internal field transfers this
frequency dependence to the higher-order spatial and mixed derivatives
of the Love--Rosenau equation.

The solitary pulse studied below is an exact solution of the
conservative limiting equation. It is not, solely by virtue of the
singular approximation, an exact solitary wave of the full system for
\(\tau_\sigma>0\). For a reversible parent system
\(\bM=0\), a nearby coherent pulse may persist for small
\(\tau_\sigma\), but persistence of an exact homoclinic orbit requires
a separate travelling-wave analysis of the full system.

The distinction becomes sharper in the dissipative case. Let a
localized constant-shape travelling wave of the full system exist on
the real line and suppose that \(\bM\) is positive definite. After
integration of the supplementary law, translation invariance of the
parent-system mechanical energy and vanishing boundary fluxes give
\begin{equation}
 0
 =
 \int_{\mathbb{R}}\Sigma\,\dd z
 =
 -\int_{\mathbb{R}}
 (\sigma,\chi)\,\bM\,(\sigma,\chi)^T\,\dd z.
 \label{eq:dissipative-travelling-energy}
\end{equation}
Consequently \(\sigma=\chi=0\) along the wave, and the non-trivial
dispersive pulse considered here is excluded. If \(\bM\) is only
positive semidefinite, an exact pulse can survive this argument only
inside its null space. With weak dissipation,
\(\bM=\varepsilon\bM_0\), and small fast inertia
\(\tau_\sigma=O(\varepsilon)\), the conservative solitary pulse may
instead be metastable: it can propagate for long but finite times while
its amplitude and width evolve slowly. Thus the dispersive reduction is
a finite-time approximation to the slow dynamics, not a claim that an
eternal conservative soliton persists in a genuinely dissipative
hyperbolic parent system.

These alternatives are tested numerically in
Section~\ref{subsec:numerical-persistence}. The full parent system is
initialized with the reduced pulse, while \(\sigma\) and \(Q\) are
reconstructed from the \(\tau_\sigma=0\) travelling-wave relations.
The computations compare decreasing positive values of
\(\tau_\sigma\), monitor phase and shape errors in the reversible case,
and measure the slow decay of amplitude and parent-system mechanical energy when
\(\bM=\varepsilon\bM_0\).
\end{remark}

The positivity of the higher-mode energy implies
\begin{equation}
 \beta-\frac{\rho}{\mu}\alpha
 =
 \frac{\tau_Q}{\kappa^2}>0.
 \label{eq:alpha-beta-constraint}
\end{equation}
Equivalently,
\begin{equation}
 \sqrt{\frac{\alpha}{\beta}}
 <
 \sqrt{\frac{\mu}{\rho}}.
 \label{eq:speed-inequality}
\end{equation}

The nonlinear dispersive term is therefore not introduced independently:
it is generated by the same constitutive stress \(T(F)=W'(F)\) that
governs equilibrium elasticity.

\begin{remark}[Modified KdV long-wave limit]
A standard weakly nonlinear long-wave scaling gives an additional
consistency check on the reduced equation. Set \(U=u_z\), take
\[
 T(U)=\mu U+\mu_nU^3,
\]
and introduce
\[
 X=\varepsilon^{1/2}(z-c_0t),
 \qquad
 \tau=\varepsilon^{3/2}t,
 \qquad
 U=\varepsilon^{1/2}U_1+O(\varepsilon^{3/2}).
\]
The leading order gives \(c_0^2=\mu/\rho\). At the next order, after
one integration in \(X\) for localized disturbances, one obtains the
modified Korteweg--de Vries equation
\begin{equation}
 (U_1)_\tau
 +\frac{3\mu_n}{2\rho c_0}U_1^2(U_1)_X
 +\frac{\tau_Qc_0}{2\rho\kappa^2}(U_1)_{XXX}
 =0.
 \label{eq:mkdv-limit}
\end{equation}
Indeed, on the acoustic cone \(c_0^2=\mu/\rho\), the leading
contributions of the fourth-order spatial term and of the mixed
space--time term combine according to
\[
 \frac{(\rho a^2+\tau_Q)c_0^2-a^2\mu}
      {2\rho c_0\kappa^2}
 =\frac{\tau_Qc_0}{2\rho\kappa^2}>0.
\]
Thus the same positive higher-mode inertia that yields the strict
subcharacteristic condition also supplies the leading mKdV dispersion.
This asymptotic observation is not used in the exact travelling-wave
analysis below.
\end{remark}

The analysis of localized travelling waves, including the question of
compactons, requires additional existence and regularity arguments.

\section{Energy identity for the reduced equation}

We now analyze the reduced equation independently of the parent
variables.  In the notation of Section~\ref{sec:nonlinear-extension},
consider
\begin{equation}
 \rho u_{tt}-[T(u_z)]_z+\alpha u_{zzzz}-\beta u_{zztt}=0,
 \qquad \alpha,\beta>0.
 \label{eq:target}
\end{equation} Set
\begin{equation}
 p=F_z=u_{zz},
 \qquad
 q=v_z=u_{zt}.
 \label{eq:pq}
\end{equation}
Then
\begin{equation}
 p_t-q_z=0.
 \label{eq:higher-kinematic}
\end{equation}
Equation \eqref{eq:target} can be written as
\begin{equation}
 \partial_t(\rho v-\beta q_z)
 -
 \partial_z[T(F)-\alpha p_z]
 =0.
 \label{eq:generalized-momentum}
\end{equation}

\begin{proposition}[Energy identity for the reduced equation]
\label{prop:target-energy}
Every smooth solution of \eqref{eq:target} satisfies
\begin{equation}
 \partial_t h^0_{\rm LR}
 +
 \partial_z h^1_{\rm LR}
 =0,
 \label{eq:LR-energy-law}
\end{equation}
where
\begin{equation}
 h^0_{\rm LR}
 =
 \frac{\rho}{2}v^2
 +
 W(F)
 +
 \frac{\alpha}{2}p^2
 +
 \frac{\beta}{2}q^2,
 \label{eq:LR-energy}
\end{equation}
with \(T(F)=W'(F)\), and
\begin{equation}
 h_{\rm LR}
 =
 -vT(F)
 -
 \alpha(pq-vp_z)
 -
 \beta vq_t.
 \label{eq:LR-flux}
\end{equation}
\end{proposition}

\begin{proof}
Using \(F_t=q\) and \(p_t=q_z\),
\[
 (h^0_{\rm LR})_t
 =
 \rho vv_t+T(F)q+\alpha p q_z+\beta q q_t.
\]
Multiplication of \eqref{eq:target} by \(v\) gives
\[
 \rho vv_t
 =
 (vT)_z-Tq
 -\alpha v p_{zz}
 +\beta v(q_t)_z.
\]
Finally,
\[
 -vp_{zz}+pq_z=(pq-vp_z)_z,
 \qquad
 v(q_t)_z+qq_t=(vq_t)_z,
\]
which proves \eqref{eq:LR-energy-law}.
\end{proof}

\begin{remark}[Reduced flux and interstitial working]
For the gradient part of the reduced energy,
\[
 W_{\rm grad}(F,p)=W(F)+\frac{\alpha}{2}p^2,
 \qquad p=F_z,
\]
the one-dimensional hyperstress is
\[
 m=\frac{\partial W_{\rm grad}}{\partial p}=\alpha p.
\]
Since \(q=F_t=v_z\), the Dunn--Serrin interstitial-work flux is,
up to the sign convention used for the total energy flux,
\begin{equation}
 w_{\rm DS}=m q=\alpha p q.
 \label{eq:DS-interstitial-work}
\end{equation}
Accordingly, the reduced flux in \eqref{eq:LR-flux} may be rewritten as
\begin{equation}
 h^1_{\rm LR}
 =-v\bigl[T(F)-\alpha p_z\bigr]
 -w_{\rm DS}
 -\beta vq_t.
 \label{eq:LR-flux-decomposition}
\end{equation}
The first term is the mechanical power associated with the effective
gradient stress, the second is precisely the interstitial-working
contribution, and the last term is the microinertial correction
associated with the mixed derivative. This is the reduced
one-dimensional structure discussed by Dunn and Serrin
\cite{DunnSerrin1985}.

The conceptual distinction from the parent RET description remains
important. At the parent level the state is local and the internal
energy flux is \(k=-a\sigma Q\); no gradient-dependent work flux is
postulated. The interstitial-working term appears only after the
internal fields have been eliminated and the reduced energy identity
has been rearranged by integration by parts. Equation
\eqref{eq:LR-flux-decomposition} is therefore an exact identification
at the reduced level, but it is not a pointwise identity between
\(k\) and \(w_{\rm DS}\), nor do we require here a separate theorem on
asymptotic convergence of the two flux representations.
\end{remark}

The benchmark \eqref{eq:LR-energy} is strictly convex when
\begin{equation}
 \rho>0,\qquad
 W''(F)>0,\qquad
 \alpha>0,\qquad
 \beta>0.
 \label{eq:benchmark-convexity}
\end{equation}

\section{Solitary-wave consequences of the nonlinear reduction}
\label{sec:solitary-waves}

\subsection{Analytical solitary pulses}
\label{subsec:analytical-solitary-pulses}

The construction above is independent of any particular travelling
wave.  Nevertheless, the reduced nonlinear equation has a useful
first integral that permits a preliminary analytical description of
localized pulses.  This subsection records that consequence without
using it as a premise of the RET construction.  It also corrects the
sign obstruction that excludes the compactly supported cosine profile
initially considered in exploratory calculations.  Related solitary
and compact-like waves in phenomenological dispersive elasticity were
studied in \cite{DestradeSaccomandi2006,DestradeSaccomandi2008}.

Differentiate \eqref{eq:nonlinear-LR-limit} with respect to \(z\), set
\(U=u_z\), and seek a travelling wave \(U=U(\xi)\),
\(\xi=z-ct\).  After two integrations, under the localized conditions
\(U,U'\to0\) as \(|\xi|\to\infty\), one obtains
\begin{equation}
 \rho c^2U-T(U)
 +\frac{a^2}{\kappa^2}[T(U)]''
 -\frac{(\rho a^2+\tau_Q)c^2}{\kappa^2}U''=0.
 \label{eq:TW-integrated}
\end{equation}
Define
\begin{equation}
 \mathcal A(U)
 =
 \frac{a^2T'(U)-(\rho a^2+\tau_Q)c^2}{\kappa^2}.
 \label{eq:A-travelling}
\end{equation}
Then \eqref{eq:TW-integrated} is
\begin{equation}
 \mathcal A(U)U''+\mathcal A'(U)(U')^2
 +\rho c^2U-T(U)=0,
 \label{eq:TW-A}
\end{equation}
and multiplication by \(\mathcal A(U)U'\) gives the exact first
integral
\begin{equation}
 \boxed{
 \frac12\mathcal A(U)^2(U')^2
 +\int_0^U\bigl(\rho c^2s-T(s)\bigr)\mathcal A(s)\,\dd s=C.
 }
 \label{eq:TW-first-integral}
\end{equation}

For the first nonlinear elastic correction
\begin{equation}
 T(U)=\mu U+\mu_nU^3,
 \qquad \mu>0,
 \label{eq:cubic-stress-wave}
\end{equation}
write
\begin{equation}
 \mathcal A(U)=A_0+A_2U^2,
 \qquad
 A_0=\frac{a^2\mu-(\rho a^2+\tau_Q)c^2}{\kappa^2},
 \qquad
 A_2=\frac{3a^2\mu_n}{\kappa^2}.
 \label{eq:A0A2-wave}
\end{equation}
For a pulse, \(C=0\), and \eqref{eq:TW-first-integral} becomes
\begin{equation}
 \frac12(A_0+A_2U^2)^2(U')^2
 +C_2U^2+C_4U^4+C_6U^6=0,
 \label{eq:pulse-energy}
\end{equation}
where
\begin{subequations}
\label{eq:C-wave}
\begin{align}
 C_2&=\frac{A_0(\rho c^2-\mu)}{2},\\
 C_4&=\frac{A_2(\rho c^2-\mu)-A_0\mu_n}{4},\\
 C_6&=-\frac{A_2\mu_n}{6}
     =-\frac{a^2\mu_n^2}{2\kappa^2}.
\end{align}
\end{subequations}

\begin{proposition}[Smooth solitary pulse]
\label{prop:smooth-pulse}
Assume
\begin{equation}
 C_2<0,
 \qquad C_4>0,
 \qquad
 \Delta:=C_4^2-4C_2C_6>0,
 \label{eq:pulse-conditions}
\end{equation}
and assume that \(A_0+A_2U^2\) does not vanish between the rest state
and the first turning point.  Set
\begin{equation}
 D=\sqrt{\Delta},
 \qquad
 U_m^2=\frac{-2C_2}{C_4+D}.
 \label{eq:pulse-amplitude}
\end{equation}
Then \eqref{eq:pulse-energy} admits an even smooth homoclinic pulse,
unique up to translation and sign, with maximum amplitude \(U_m\) and
exponential decay at infinity.  An exact parametric representation is
\begin{subequations}
\label{eq:parametric-pulse}
\begin{align}
 U^2(x)
 &=\frac{-2C_2}
 {C_4+D\cosh\!\left(2\sqrt{-2C_2}\,(x-x_0)\right)},
 \label{eq:U-parametric}\\
 \xi(x)
 &=\xi_0+A_0(x-x_0)
 +\frac{A_2}{\sqrt{-2C_6}}
 \operatorname{artanh}\!\left[
 \sqrt{\frac{C_4-D}{C_4+D}}
 \tanh\!\left(\sqrt{-2C_2}\,(x-x_0)\right)
 \right].
 \label{eq:xi-parametric}
\end{align}
\end{subequations}
\end{proposition}

\begin{proof}
Introduce the auxiliary coordinate by
\begin{equation}
 \frac{\dd\xi}{\dd x}=A_0+A_2U^2.
 \label{eq:hodograph-coordinate}
\end{equation}
Equation \eqref{eq:pulse-energy} then reduces to
\begin{equation}
 \frac12U_x^2+C_2U^2+C_4U^4+C_6U^6=0.
 \label{eq:standard-pulse-energy}
\end{equation}
With \(W=U^2\), equation
\eqref{eq:standard-pulse-energy} becomes
\[
 (W_x)^2
 =
 -8W^2\bigl(C_2+C_4W+C_6W^2\bigr),
\]
whose homoclinic integral is precisely
\eqref{eq:U-parametric}. Integrating
\(\dd\xi/\dd x=A_0+A_2W(x)\) gives
\eqref{eq:xi-parametric}.  The assumptions guarantee that the first
positive zero of \(C_2+C_4W+C_6W^2\) is \(W=U_m^2\) and that the
coordinate map is monotone.  Linearization of
\eqref{eq:pulse-energy} at \(U=0\) yields
\begin{equation}
 U(\xi)\sim
 \exp\!\left[-\frac{\sqrt{-2C_2}}{|A_0|}|\xi|\right],
 \qquad |\xi|\to\infty.
 \label{eq:pulse-tail}
\end{equation}
\end{proof}

\begin{figure}[t]
 \centering
 \includegraphics[width=0.72\textwidth]{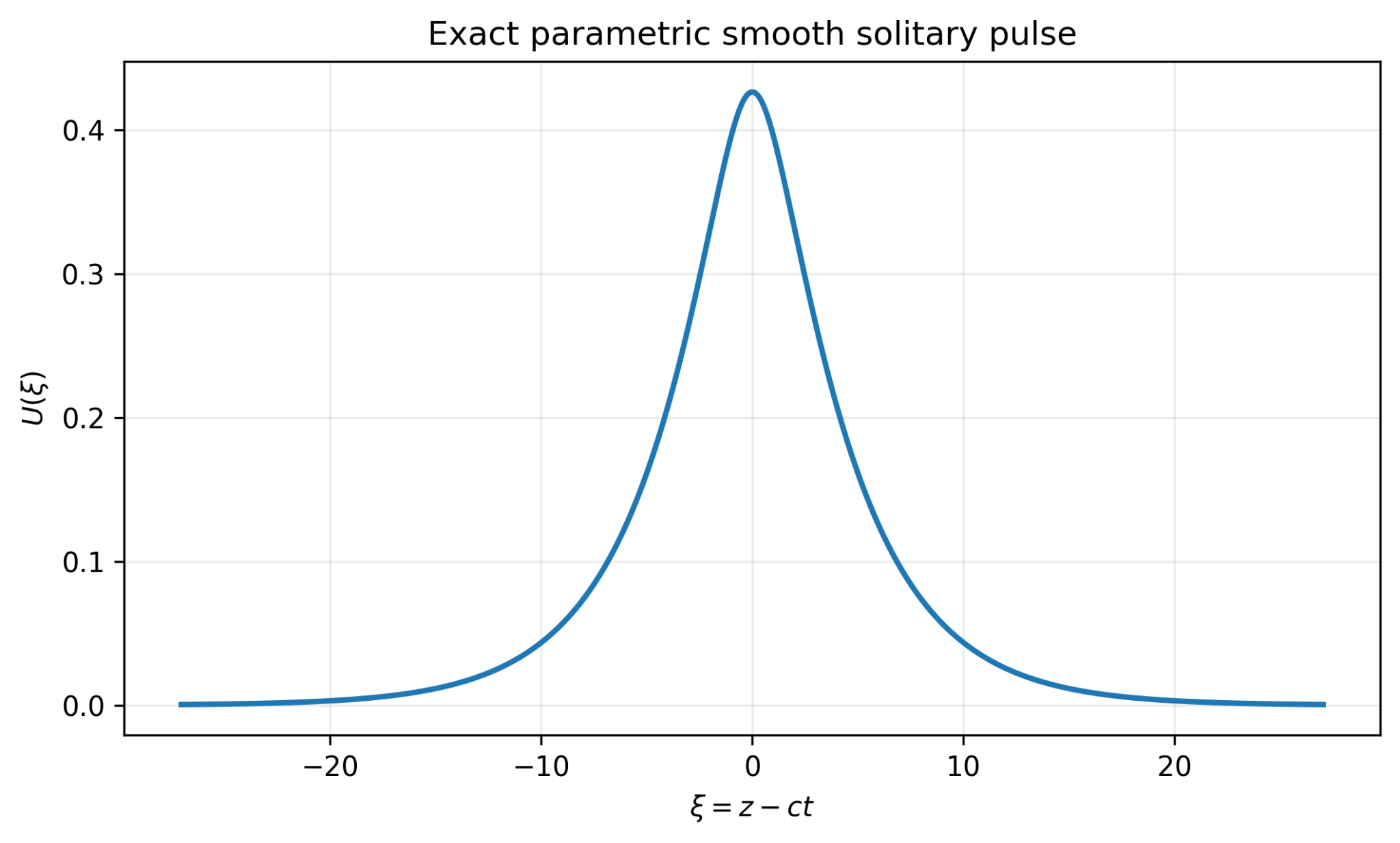}
 \caption{Representative exact parametric smooth solitary pulse from
 Proposition~\ref{prop:smooth-pulse}. The parameters are
 \(\rho=\mu=\mu_n=a=\kappa=\tau_Q=1\) and \(c=1.04\). They give
 \(C_2=-0.04746\), \(C_4=0.35200\), \(C_6=-0.5\),
 \(\Delta=0.02899\), and \(U_m=0.42632\). The curve is obtained
 directly from the parametric representation
 \eqref{eq:parametric-pulse}, not from a time-dependent numerical
 integration.}
 \label{fig:smooth-solitary-pulse}
\end{figure}

The velocity restriction can be expressed in physical parameters.
Define
\begin{equation}
 c_0=\sqrt{\frac{\mu}{\rho}},
 \qquad
 c_{\rm micro}
 =\sqrt{\frac{a^2\mu}{\rho a^2+\tau_Q}}<c_0.
 \label{eq:wave-speeds}
\end{equation}
Then
\begin{equation}
 C_2=
 \frac{\rho(\rho a^2+\tau_Q)}{2\kappa^2}
 (c^2-c_0^2)(c_{\rm micro}^2-c^2).
 \label{eq:C2-factorized}
\end{equation}
Consequently, the interval
\(c_{\rm micro}<c<c_0\) gives \(C_2>0\), not \(C_2<0\).
For a hardening correction \(\mu_n>0\), the smooth-pulse branch is
therefore supersonic. The explicit branch analysed here is restricted
to this hardening case. If \(\mu_n<0\), then \(A_2\) changes sign and
the geometry of both the effective energy polynomial and the
hodograph map is different; possible softening branches require a
separate classification and are not considered in this paper.
Moreover,
\begin{equation}
 \Delta
 =-\frac{c^2\mu_n^2\tau_Q}{16\kappa^4}
 \left[(8\rho a^2-\tau_Q)c^2-8a^2\mu\right].
 \label{eq:Delta-physical}
\end{equation}
If \(0<\tau_Q<8\rho a^2\), a sufficient and explicit window is
\begin{equation}
 c_0<c<
 \sqrt{\frac{8a^2\mu}{8\rho a^2-\tau_Q}}.
 \label{eq:supersonic-window}
\end{equation}
For \(\tau_Q\geq8\rho a^2\), the discriminant condition holds for
every \(c>c_0\).  On this branch the non-vanishing condition for
\(\mathcal A\) up to the first turning point follows from the same
inequality.

When the distortion of the physical coordinate is weak,
\(|A_2|U^2\ll|A_0|\), one may use \(\xi\simeq A_0x\) in
\eqref{eq:U-parametric}.  This gives a convenient quasi-explicit pulse
profile, but the exact solution is the parametric pair
\eqref{eq:parametric-pulse}.

\begin{remark}[Degeneracy and the absence of a compacton]
When \(A_0=0\), one also has \(C_2=0\), and the first integral
\eqref{eq:pulse-energy} takes the degenerate form
\begin{equation}
 U^4\left[
 (U')^2-
 \frac{\kappa^2}{9a^2\mu_n}
 \left(\frac{3(\mu-\rho c^2)}{2}+\mu_nU^2\right)
 \right]=0.
 \label{eq:degenerate-pulse-integral}
\end{equation}
The factor \(U^4\) shows the loss of Lipschitz regularity at the rest
state, a mechanism that is often relevant to compactification
\cite{CirilloSaccomandiSciarra2019}. Degeneracy alone, however, is not
sufficient to produce a compacton.

Indeed, \(A_0=0\) and the RET convexity condition \(\tau_Q>0\) imply
\begin{equation}
 c^2=\frac{a^2\mu}{\rho a^2+\tau_Q},
 \qquad
 \mu-\rho c^2
 =\frac{\mu\tau_Q}{\rho a^2+\tau_Q}>0.
 \label{eq:degenerate-speed-sign}
\end{equation}
For the hardening case \(\mu_n>0\), the non-trivial branches are
\begin{equation}
 U(\xi)=
 \sqrt{\frac{3(\mu-\rho c^2)}{2\mu_n}}
 \sinh\!\left[\pm\frac{\kappa}{3|a|}(\xi-\xi_0)\right],
 \label{eq:sinh-branch}
\end{equation}
and are not localized. For the degenerate softening case
\(\mu_n=-\nu<0\), \(\nu>0\), equation
\eqref{eq:degenerate-pulse-integral} instead becomes
\begin{equation}
 (U')^2
 =\frac{\kappa^2}{9a^2}
 \left(U^2-\frac{3(\mu-\rho c^2)}{2\nu}\right),
 \label{eq:degenerate-softening}
\end{equation}
whose real non-constant branches are of hyperbolic-cosine type and
never reach \(U=0\). In neither case does a non-trivial branch possess
two finite zeros that could be joined to the rest state. The
truncated-cosine compacton is therefore excluded not merely for the
hardening example, but throughout this degenerate quadratic-gradient
mechanism compatible with the present RET reduction. The general
non-degenerate softening problem remains outside the present
classification.
\end{remark}

\subsection{Numerical persistence in the hyperbolic parent system}
\label{subsec:numerical-persistence}

We now test whether the exact pulse of the singular reduced equation
provides a finite-time approximation to the full first-order
hyperbolic hierarchy. We retain
\[
 \rho=\mu=\mu_n=a=\kappa=\tau_Q=1,
 \qquad c=1.04,
\]
as in Figure~\ref{fig:smooth-solitary-pulse}, and solve the parent
system on the periodic interval \([-80,80]\). In addition to the
reversible production \eqref{eq:antisymmetric-production}, we consider
the weakly dissipative choice
\begin{equation}
 \cP_1=-\kappa Q-\eta\sigma,
 \qquad
 \cP_2=\kappa\sigma-\eta Q,
 \qquad
 \eta\geq0,
 \label{eq:numerical-dissipative-production}
\end{equation}
which corresponds to \(\bM=\eta\bm I\). The reversible case is
\(\eta=0\), while in the weakly dissipative runs we set
\(\eta=\tau_\sigma\).

The initial data are reconstructed from the exact reduced pulse
\(U\) according to
\begin{equation}
 F=U,
 \qquad
 v=-cU,
 \qquad
 \sigma=\rho c^2U-T(U),
 \qquad
 (\kappa-a\partial_z)Q=v_z.
 \label{eq:numerical-initial-data}
\end{equation}
Spatial derivatives are approximated by a Fourier pseudospectral
method with \(2/3\) dealiasing. Time integration uses a second-order
IMEX--BDF2 scheme: the complete constant-coefficient linear operator,
including the terms proportional to \(1/\tau_\sigma\), is treated
implicitly, whereas only the cubic stress contribution is explicit.
The computations reported below use \(N=1024\) Fourier nodes,
\(\Delta t=5\times10^{-3}\), and final time \(t=50\).

To separate deformation of the pulse from its translation, let
\(F_0\) be the initial strain profile and define the aligned relative
shape error
\begin{equation}
 e_{\rm sh}(t)
 =
 \min_{s\in\mathbb R}
 \frac{\|F(\cdot,t)-F_0(\cdot-s)\|_{L^2}}
      {\|F_0\|_{L^2}}.
 \label{eq:numerical-shape-error}
\end{equation}
The numerical results are summarized in Table~\ref{tab:numerical-persistence}.

\begin{table}[b]
\centering
\caption{Long-time numerical evolution of the reduced solitary pulse
under the full hyperbolic parent system. Here
\(A(t)=\max_z F(z,t)\), and \(E(t)\) is the parent-system
mechanical energy. In the dissipative rows \(\eta=\tau_\sigma\); in the
reversible rows \(\eta=0\).}
\label{tab:numerical-persistence}
\begin{tabular}{ccccc}
\hline
case & \(\tau_\sigma\) & \(A(50)/A(0)\) & \(E(50)/E(0)\) & \(e_{\rm sh}(50)\)\\
\hline
reversible  & 0.100 & 0.986638 & 0.999999 & 0.013461\\
dissipative & 0.100 & 0.840289 & 0.871075 & 0.115639\\
reversible  & 0.050 & 0.992772 & 0.999999 & 0.006838\\
dissipative & 0.050 & 0.907787 & 0.927137 & 0.064279\\
reversible  & 0.025 & 0.996332 & 0.999999 & 0.003432\\
dissipative & 0.025 & 0.949675 & 0.960990 & 0.034066\\
\hline
\end{tabular}
\end{table}

In the reversible hierarchy the parent-system mechanical energy is conserved to about
six significant digits, and the aligned shape error decreases nearly
linearly with \(\tau_\sigma\) over the tested range. Thus the exact
pulse of the reduced equation provides an accurate finite-time
approximation to a coherent pulse of the hyperbolic parent system. In
the dissipative hierarchy the pulse remains recognizable but loses
amplitude and energy slowly; the decay decreases systematically with
\(\eta=\tau_\sigma\), in agreement with the metastable interpretation
of Remark~\ref{rem:hyperbolic-dispersive-pulse}.

\begin{figure}[!t]
 \centering
 \includegraphics[width=0.64\textwidth]{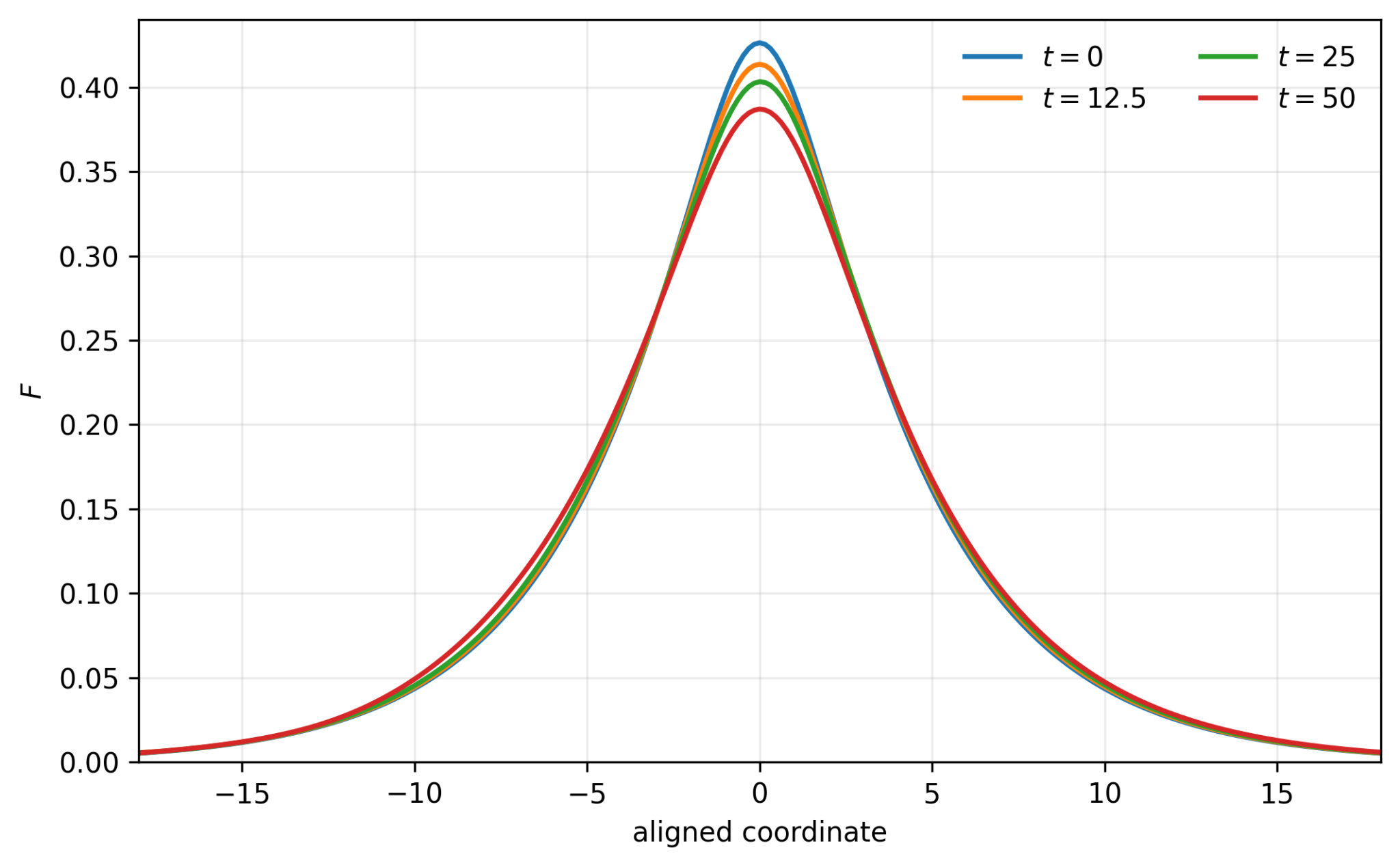}\\[0.35em]
 \includegraphics[width=0.64\textwidth]{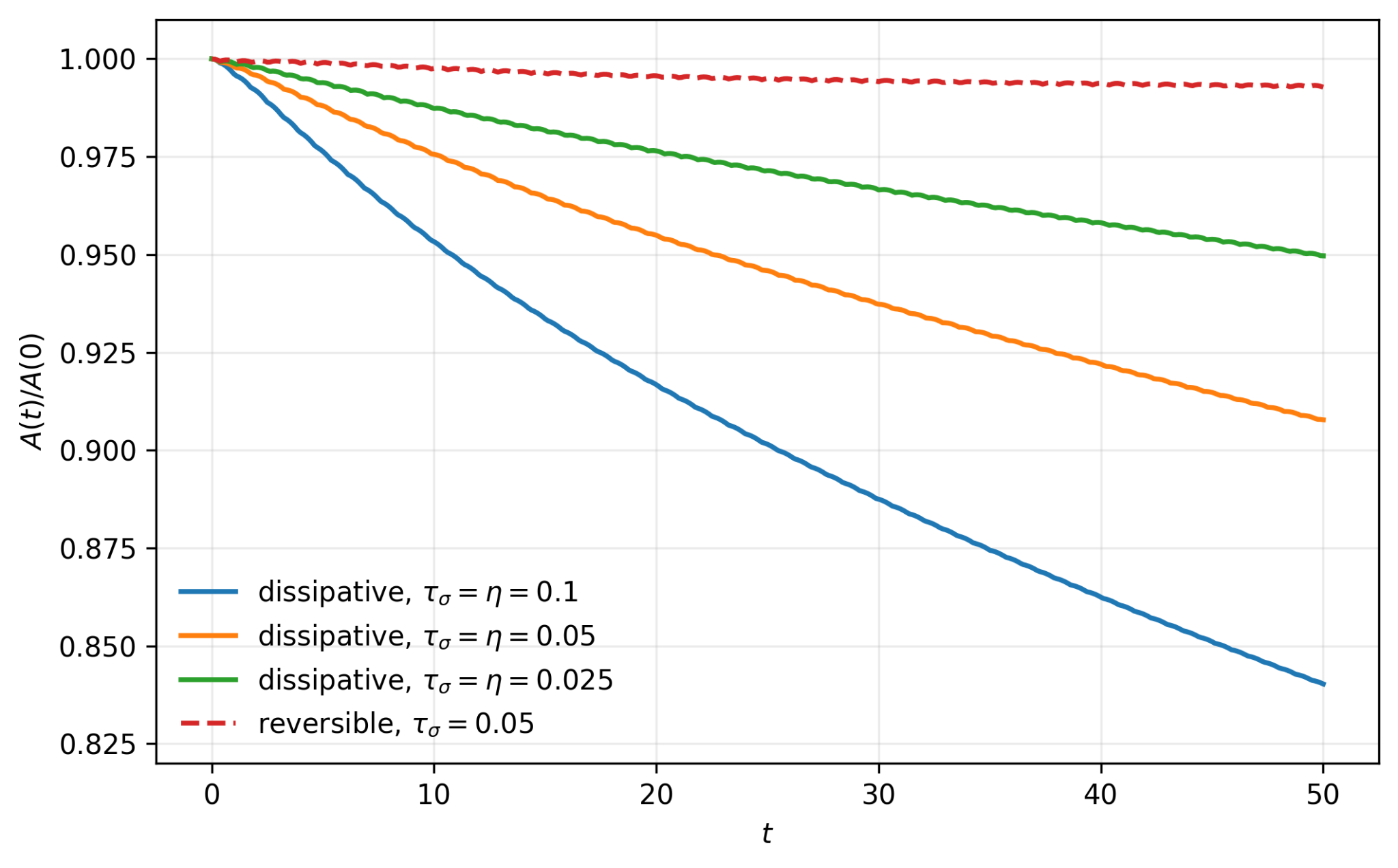}
 \caption{Long-time numerical persistence of the reduced pulse in the
 full hyperbolic parent system. Upper panel: aligned strain profiles in
 the weakly dissipative case \(\tau_\sigma=\eta=0.05\), showing the
 reduction of the peak without obscuring it by translation. Lower
 panel: normalized peak amplitude for three weakly dissipative runs
 and for the representative reversible run
 \(\tau_\sigma=0.05\).}
 \label{fig:numerical-persistence}
\end{figure}

A resolution check was also performed for \(\tau_\sigma=0.05\) at
\(t=50\), using \((N,\Delta t)=(512,10^{-2})\),
\((1024,5\times10^{-3})\), and \((2048,2.5\times10^{-3})\). The final
aligned shape errors were respectively
\(6.860\times10^{-3}\), \(6.838\times10^{-3}\), and
\(6.832\times10^{-3}\) in the reversible case, and
\(6.429\times10^{-2}\), \(6.428\times10^{-2}\), and
\(6.428\times10^{-2}\) in the dissipative case. The observed
persistence and decay are therefore insensitive to this refinement.
These computations are not a stability proof: spectral or orbital
stability of the pulse remains an open analytical problem.

\section{Discussion and conclusions}

The purpose of this work is to construct dispersive elasticity from a
local balance-law hierarchy rather than to postulate higher spatial
gradients as primitive constitutive quantities.  The general
Ruggeri--Strumia representation determines the main field, the total
stress, the internal fluxes and the admissible production directly
from the supplementary energy law.  The canonical choice
\(
 \psi_1=Z(\sigma)-F
\),
\(
 \psi_2=Y(Q)
\)
contains nonlinear elasticity and the one-field generalized-stress
model as exact principal subsystems.

The linear theory gives the central realizability result.  Convexity
requires positive internal inertias, in particular \(\tau_Q>0\), and
the fast-stress limit generates the Love--Rosenau equation exactly.
The inequality
\(
 \beta>(\rho/\mu)\alpha
\)
is both necessary and sufficient for a genuinely coupled convex
reversible two-field realization, and the canonical hierarchy spans
the entire admissible coefficient domain.

The methodological distinction between balance laws and constitutive
relations is essential.  Following the standpoint developed by
Ruggeri, objectivity is imposed on constitutive mappings but not on the
balance equations from which non-local parabolic laws may arise as
singular reductions.  The parent hierarchy is local and compatible
with Galilean transformations between inertial observers; no objective
time derivative is introduced solely to enforce objectivity of an
independent balance field.

The nonlinear RET architecture accommodates general elastic stresses
and non-quadratic higher-field energies. For the exact nonlinear
Love--Rosenau reduction studied here, however, we return to quadratic
higher-field inertia while retaining a general nonlinear elastic
stress. The resulting travelling-wave first integral proves the
existence of a smooth
supersonic solitary-pulse branch for the first cubic elastic
correction, provides an exact parametric profile, and shows that the
degenerate choice previously associated with a truncated cosine leads
instead to a non-localized hyperbolic-sine branch.  The solitary-wave
analysis is therefore a consequence of the RET reduction, not a
numerical premise for it. Direct simulations of the full parent
hierarchy further show that the reduced pulse persists over finite
times in the reversible small-inertia regime, whereas weak dissipation
produces a slowly decaying metastable pulse. The numerical evidence is
consistent with the singular-limit interpretation but is not a claim
of spectral or orbital stability.

The additional flux \(k\) has a clear local meaning in the parent
system as energy transport between internal levels. A term of
interstitial-working type appears only in the reduced higher-gradient
energy identity and is not needed for the parent RET construction.
Multidimensional Galilean-covariant extensions, a complete
classification of periodic and kink waves, the softening case, and
the spectral or orbital stability of the solitary pulse remain open.

\appendix

\section{Proof of the supplementary-law compatibility theorem}

We start from the general, still undetermined, row main field
\begin{equation}
 \bu'
 =
 \bigl(\zeta,\lambda,\Lambda_1,\Lambda_2\bigr),
 \label{eq:main-components}
\end{equation}
with
\[
 \bLambda=\bigl(\Lambda_1,\Lambda_2\bigr).
\]
From \eqref{eq:density-vector},
\begin{equation}
 \dd\bu
 =
 \begin{pmatrix}
  \rho\,\dd v\\[1mm]
  \dd F\\[1mm]
  \dd\bpsi
 \end{pmatrix},
 \qquad
 \dd\bpsi
 =
 \bpsi_F\,\dd F
 +
 \bJ
 \begin{pmatrix}
  \dd P\\
  \dd Q
 \end{pmatrix}.
 \label{eq:du}
\end{equation}
On the other hand,
\begin{equation}
 \dd h^0
 =
 \rho v\,\dd v
 +
 \Psi_F\,\dd F
 +
 \be^T
 \begin{pmatrix}
  \dd P\\
  \dd Q
 \end{pmatrix}.
 \label{eq:dh0}
\end{equation}
The temporal Ruggeri--Strumia identity
\begin{equation}
 \dd h^0=\bu'\,\dd\bu
 \label{eq:temporal-RS-proof}
\end{equation}
gives, by comparison of the coefficients of
\(\dd v,\dd F,\dd P,\dd Q\),
\begin{equation}
 \rho\zeta=\rho v,
 \qquad
 \lambda+\bLambda\bpsi_F=\Psi_F,
 \qquad
 \bLambda\bJ=\be^T.
 \label{eq:component-comparison-h0}
\end{equation}
Since \(\rho>0\) and \(\bJ\) is invertible, it follows that
\begin{equation}
 \zeta=v,
 \qquad
 \bLambda=\be^T\bJ^{-1},
 \qquad
 \lambda=\Psi_F-\bLambda\bpsi_F.
 \label{eq:main-solution}
\end{equation}

We next use the spatial Ruggeri--Strumia identity. From the postulated
energy flux \eqref{eq:general-energy-flux},
\begin{equation}
 \dd h^1
 =
 -\cT\,\dd v
 -v\,\dd\cT
 +\dd k.
 \label{eq:dh1-postulated}
\end{equation}
On the other hand, from \eqref{eq:flux-production},
\begin{align}
 \dd h^1
 &=\bu'\,\dd\bF \notag\\
 &=-\zeta\,\dd\cT-\lambda\,\dd v
 +\Lambda_1\,\dd\Omega_1
 +\Lambda_2\,\dd\Omega_2.
 \label{eq:dh-oneform}
\end{align}
Using \(\zeta=v\), comparison with
\eqref{eq:dh1-postulated} yields
\begin{equation}
 \lambda=\cT
 \label{eq:lambda-stress}
\end{equation}
and
\begin{equation}
 \dd k
 =
 \Lambda_1\,\dd\Omega_1
 +
 \Lambda_2\,\dd\Omega_2.
 \label{eq:flux-exactness-proof}
\end{equation}
Combining \eqref{eq:lambda-stress} with
\eqref{eq:main-solution} gives
\begin{equation}
 \cT
 =
 \Psi_F-\bLambda\bpsi_F,
\end{equation}
and therefore
\begin{equation}
 \bu'
 =
 \bigl(v,\cT,\Lambda_1,\Lambda_2\bigr),
 \qquad
 \bLambda
 =
 \bigl(\Lambda_1,\Lambda_2\bigr)
 =
 \be^T\bJ^{-1}.
\end{equation}
Finally,
\begin{equation}
 \Sigma
 =
 \bu'\,\bP
 =
 \Lambda_1\cP_1+\Lambda_2\cP_2,
\end{equation}
and the requirement \(\Sigma\leq0\) completes the proof.

\hfill\qedsymbol

\clearpage

\section*{Acknowledgements}
This work has been carried out in the framework of activities of the
National Group of Mathematical Physics (GNFM, INdAM).

\section*{Data accessibility}
No external datasets were used. Figure~\ref{fig:smooth-solitary-pulse} is obtained directly from the analytical parametric formula in Proposition~\ref{prop:smooth-pulse}. The Python code and numerical output underlying Figure~\ref{fig:numerical-persistence} and Table~\ref{tab:numerical-persistence} will be deposited in a public repository; the permanent repository link and DOI will be inserted in the final submitted version.

\section*{Author contributions}
T.R. and G.S. contributed to the conception of the model, the mathematical analysis and the interpretation of the results. Both authors drafted, revised and approved the manuscript.

\section*{Competing interests}
We declare we have no competing interests.

\section*{Funding}
This research received no specific grant from any funding agency in the public, commercial or not-for-profit sectors.

\clearpage

% References are kept directly in this file to avoid dependence on an
% external .bbl or .bst file in Overleaf.  The entries follow the
% Royal Society numerical style.

\end{document}